\documentclass[]{elsarticle}

\usepackage{comment}

\usepackage{url}
\usepackage{wrapfig}

\usepackage{subfig}
\usepackage{amsmath}

\usepackage{amsthm}
\usepackage{amssymb}
\usepackage{wrapfig}
\usepackage{graphicx}
\usepackage{paralist}
\usepackage[active]{srcltx}
\usepackage{algorithm}
\usepackage[noend]{algorithmic}
\usepackage{array}
\usepackage{enumitem}
\usepackage{color}

\newenvironment{proof sketch}[1]{\noindent {\emph{Proof sketch of #1:}}}{\hfill }
\newtheorem{theorem}{Theorem}

\newtheorem{clm}{Claim}
\newtheorem{lemma}{Lemma}

\newtheorem{coro}{Corollary}

\newtheorem{observation}{Observation}

{\bfseries}{\itshape}

\newtheorem*{thm_necessity-matrix2}{Theorem \ref{thm:necessity-matrix2}}

\newtheorem*{definition_bal}{(i) $\alpha$-Balanced Boundary}
\newtheorem*{definition_cli}{(ii) $\beta$-Clique Emulation}
\newtheorem*{definition_con}{(iii) $\gamma$-convergecast}

\def\abs#1{\lvert #1 \rvert}

\def\P{\mathcal{P}}
\def\C{\mathcal{C}}

\def\pj{\mathbb{\textsc{PJ}}}
\def\MF{\mathbb{\textsc{MergeFrags}}}

\def\mwoe{\texttt{mwoe}}
\def\mp{\texttt{mp}}
\def\spk{\texttt{spk}}

\def\sndmsg{{\tt SendMsg}}

\def\nc{n_{\scriptscriptstyle \C}}
\def\np{n_{\scriptscriptstyle \P}}

\def\axiomBoundary{{\cal A}_{B}}
\def\axiomClique{{\cal A}_{E}}
\def\axiomConvergecast{{\cal A}_{C}}
\def\algoMatrixByVector{\text{2}}
\def\algoTranspose{\text{1}}
\def\algoMatrixByMatrix{\text{3}}
\def\algoFindRank{\text{4}}
\def\algoMedian{\text{5}}
\def\algoFindModes{\text{6}}
\def\algoFindDistinct{\text{7}}
\def\algoFindKTopAreas{\text{8}}

\def\GVECP{G(V,E,\C,\P)}
\def\cp{\left\langle \C,\P\right\rangle}

\makeatletter
\def\ps@pprintTitle{%
 \let\@oddhead\@empty
 \let\@evenhead\@empty
 \def\@oddfoot{}%
 \let\@evenfoot\@oddfoot}
\makeatother
\begin{document}

%\begin{titlepage}
\title{
%An Axiomatic Approach to Core-Periphery Social Networks
%Core-Periphery Computational Model for Social Networks
%Efficient Distributed Computation in Core-Periphery Networks
%Distributed MST in Core-Periphery Networks
Distributed Computing on Core-Periphery Networks: Axiom-based Design 
}

\newcommand*\samethanks[1][\value{footnote}]{\footnotemark[#1]}

%\author{
%Chen Avin \inst{1}\fnmsep\thanks{Part of this work was done while the author was a visitor at ICERM, Brown university.}\fnmsep\thanks{Supported in part by the Israel Science Foundation (grant 1549/13).}
%\and
%Michael Borokhovich \inst{1}\fnmsep\samethanks
%\and
%Zvi Lotker \inst{1}\fnmsep\samethanks
%\and
%David Peleg \inst{2}\fnmsep\samethanks
%}

\author[bgu]{Chen Avin\fnref{fn1,fn2}}
\ead{avin@cse.bgu.ac.il}

\author[bgu]{Michael Borokhovich\fnref{fn2}}
\ead{borokhom@cse.bgu.ac.il}

\author[bgu]{Zvi Lotker\fnref{fn2}}
\ead{zvilo@cse.bgu.ac.il}

\author[weiz]{David Peleg\fnref{fn2}}
\ead{david.peleg@weizmann.ac.il}

\fntext[fn1]{Part of this work was done while the author was a visitor at ICERM, Brown university.}
\fntext[fn2]{Supported in part by the Israel Science Foundation (grant 1549/13).}

\address[bgu]{Ben-Gurion University of the Negev, Israel.}
\address[weiz]{The Weizmann Institute, Israel.}

%\institute{
%Ben-Gurion University of the Negev, Israel. \email{\{avin,borokhom,zvilo\}@cse.bgu.ac.il}
%\and
%The Weizmann Institute, Israel. 
%\email{david.peleg@weizmann.ac.il}
%}

%\footnotetext[1]{
%\noindent
%Department of Communication Systems Engineering, Ben Gurion University of the Negev,
%Beer-Sheva, Israel.\\
%E-mails:~{\tt \{avin,borokhom,zvilo\}@cse.bgu.ac.il}.
%}
%
%\footnotetext[2]{
%\noindent
%Department of Computer Science, The Weizmann Institute,
%Rehovot, Israel.  \hbox{E-mail}:~{\tt david.peleg@weizmann.ac.il}.
%Supported in part by the Israel Science Foundation (grant 894/09),
%the United States-Israel Binational Science Foundation (grant 2008348),
%the Israel Ministry of Science and Technology (infrastructures grant),
%and the Citi Foundation.}

\begin{abstract}
% !TEX root = core_per.tex

%\begin{abstract}

%Social networks have become a prominent part of our lives and their 
%study attracts increasing attention. Intensive research is directed at 
%understanding social networks, their structure and dynamics, and developing accurate
%and useful mathematical models for them.
%
%Common approaches to modeling social networks are based on defining (mostly randomized) 
%mathematical rules for constructing a member of a given class of social networks. 
%Examples are the {\em preferential attachment} model
%and various other network models based on random graph theory.
%A complementary approach promoted here is to characterize social networks 
%and their structural components using an {\em axiomatic} approach, 
%where the axioms aim at capturing key properties of the modeled structure.
%Axiomatic-based characterizations of social networks allow us to abstract away some 
%of the less significant details. They capture and formalize the fundamental 
%properties of the structure, which affect their analysis, dynamic behavior 
%and social implications, and may assist in developing efficient algorithms 
%for various purposes. Moreover, the axiomatic approach may allow us to assess 
%the usefulness of concrete models for constructing social networks.

Inspired by social networks and complex systems, we propose a 
{\em core-periphery} network architecture that supports fast computation 
for many distributed algorithms and is robust and efficient in number of links.
Rather than providing a concrete network model, we take an axiom-based design
approach. We provide three intuitive and independent algorithmic axioms 
and prove that any network that satisfies all axioms enjoys an efficient 
algorithm for a range of tasks (such as MST, sparse matrix multiplication, 
and more).
%Moreover, the axiomatic approach allow us to 
We also show the \emph{minimality} of our axiom set: for networks that
satisfy any subset of the axioms, the same efficiency cannot be guaranteed for \emph{any} deterministic algorithm.

\end{abstract}

%\date{}

\maketitle 
%\thispagestyle{empty}

%\end{titlepage}

\section{Introduction}
% !TEX root = core_per.tex

A fundamental goal in distributed computing concerns finding network architectures that allow fast running times for various distributed algorithms, but at the same time are cost-efficient, in terms of minimizing the number of communication links between the machines and the amount of memory used by each processor.

For illustration, let's consider three basic network topologies: a star, a clique and a constant degree expander. 
The \emph{star graph} has only a linear number of links, and can compute every computable function in one round of communication. 
%(assuming links of sufficiently high capacity). 
But clearly, such an architecture has two major disadvantages:
the memory requirements of the central node do not scale, and the network is not robust (in the sense that a failure of the central node is enough to disable the network).
The \emph{complete graph}, on the other hand, is very robust, and can support extremely high performance for tasks such as information dissemination, distributed sorting and minimum spanning tree, to name a few \cite{lotker2005minimum,Lenzen:2011,Lenzen:2013:Sorting}. Also, in a complete graph, the amount of memory used by a single processor is minimal. The main drawback of that architecture is the high number of links it uses.
%, an order of $n^2$. 
\emph{Constant degree expanders} are a family of graphs that support efficient computation for many tasks. They also have linear number of links, and can effectively balance the workload between many machines. But the diameter of these graphs is lower bounded by $\Omega(\log n)$, which implies a similar lower bound on the time required for most of the interesting tasks one can consider. 

Therefore, a natural question is whether there are other candidate topologies with guaranteed good performance. We are interested in the best compromise solution: a network on which distributed algorithms have low running times, memory requirements at each node are limited, the architecture is robust to link and node failures and the total number of links is minimized (preferably linear in the number of nodes). 

To try to answer this question, we adopt in this paper an \emph{axiomatic} approach to the design of efficient networks. 
In contrast to the direct approach to network design, which is based on providing a \emph{concrete} type of networks (by deterministic or randomized construction) and showing its efficiency, the axiomatic approach attempts to abstract away the algorithmic requirements that are imposed on the concrete model.
This allows one to isolate and identify the basic requirements that a network needs for a certain type of tasks.
While there are obvious similarities between the traditional direct modeling 
approach and our axiom-based one, there are also some marked differences.
Perhaps the main difference is that
whereas the direct modeling approach focuses on rules governing the concrete {\em structural} properties of a network, our axiomatic approach relies on defining necessary \emph{operational} or \emph{algorithmic} properties required from the network, without committing to any a specific topology. 
This approach may allow us to abstract away many of the less significant details of the structure, and thus enable us to derive a simpler analysis of the essential properties of the structure under consideration, developing efficient algorithms for various purposes,
%(such as information dissemination, information extraction, computing aggregate functions etc.) 
relying solely on those abstract features. 

A closely related distinction is that while the performance of distributed algorithms is usually expressed by specific {\em structural} network parameters (e.g., diameter, degree, etc.), the axioms proposed in this work
are expressed in terms of desired \emph{algorithmic} properties 
that the network should have.

Our axiomatic approach is also influenced by the recent concept of \emph{Software Defined Networks} (SDN) \cite{Feamster:SDN}.
The concept of SDN implies layering and abstractions in the networking control plane, which allows easy configuration of overlay architectures with a specific \emph{behavior} (as opposed to specific \emph{structure}). By only defining the desired behavior, we may allow many network implementations, as long as the desired behavior is maintained. Using our axioms, we are able to define the exact abstractions required for devising efficient distributed algorithms,thus decoupling them from specific topology details.

The axioms proposed in the current work are motivated and inspired by the {\em core-periphery} structure exhibited by many social networks and complex systems.
A core-periphery network is a network structured of two distinct groups of nodes, namely, a large, sparse and weakly connected group of nodes identified as the {\em periphery}, which is loosely organized around a small, cohesive and densely connected group identified as the {\em core}.
Such a dichotomic structure appears in many domains of our life, and
has been observed in many social organizations, including modern social networks \cite{Avin2012From}. It can also be found in urban and even global systems (e.g., in global economy, the wealthiest countries constitute the core, which is highly connected by trade and transportation routes) \cite{fujita2001spatial,krugman-1991,holme2005core}. 
An analysis conducted in \cite{maccormack2010architecture} for many software systems, revealed that $75-80\%$ of the systems examined possess a core-periphery structure. 
There are also peer-to-peer networks that use a similar hierarchical structure, e.g., FastTrack \cite{Liang2006842} and Skype \cite{skype2006}, in which the supernodes can be viewed as the core, while the regular users constitute the periphery. 
Various client-server systems can also be thought of in these terms.

The vast presence of the core-periphery structure in our life suggests that it is a natural and effective design pattern. Consequently, we argue that a distributed network architecture based on it may support distributed algorithms that are
easy to devise and reason about, given that we are used to thinking in ``core-periphery terms" (for instance, in employing information processes based on collecting inputs from various peripheral sources to the core nodes, processing the data centrally, and then sending the results back to the periphery).

The main technical contribution of this paper is in proposing a minimal set of simple core-periphery-oriented axioms, and demonstrating that networks satisfying these axioms achieve efficient running time for various distributed computing tasks, while being able to maintain a linear number of edges and limited memory use. 
We identify three basic, abstract and conceptually simple (parameterized) properties, which turn out to be highly relevant to the effective interplay between core and periphery. For each of these three properties, we propose a corresponding axiom, which in our opinion captures some intuitive aspect of the desired behavior expected of a network based on a core-periphery structure. 
Let us briefly describe our three properties, along with their ``real life" 
interpretation, technical formulation and associated axioms.

The three properties are: (i) balanced boundary between the core and 
periphery, (ii) clique-like structure of the core and (iii) fast convergecast from the periphery to the core.
The first property (i) concerns the boundary between the core and the periphery.
Drawing an analogy from the world of social networks, the core can be thought of as a highly influential group, that exerts its influence (in the form of instructions, opinions, or other means) on the periphery. The mechanism through which this is done is based on certain core nodes, which are each connected to many nodes in the periphery, and act as ``\emph{ambassadors}" of the core. 
Ambassadors serve as bidirectional channels, through which information flows into the core and influence flows from the core to the periphery. However, to be effective as an ambassador, the core node must maintain a balance between its interactions with the external periphery, and its interactions with the other core members, serving as its natural ``support"; a core node that is significantly more connected to the periphery than to the core, becomes ineffective as a channel of influence. In distributed computing terms, a core node that has many connections to the periphery has to be able to distribute all the information it collected from them, to other core nodes.
Hence the relevant property is having a {\em balanced boundary}:
a set $S$ of nodes is said to have an \emph{$\alpha$-balanced} boundary if for each of its nodes, the ratio between the sizes of its neighborhoods outside and inside $S$ is at most $O(\alpha)$.
The corresponding Axiom $\axiomBoundary$ states that the core must have 
a $\Theta(1)$-balanced boundary.

The second property (ii) deals with the flow of information within the core. 
It is guided by the key observation that to be influential, the core must be 
able to accomplish fast information dissemination internally among its members.
The extreme example of a dissemination-efficient network is the complete graph,
so the core's efficiency in information flow should naturally be measured 
against that benchmark. Formally, 
a set of nodes is said to be a {\em $\beta$-clique emulator} if it can 
accomplish full communication (namely, message exchange between every pair 
of its members) in $\beta$ time (communication rounds). 
The corresponding Axiom $\axiomClique$ postulates that the core must be
a $\Theta(1)$-clique emulator. 
Note that this requirement is stronger than just requiring the core to possess
a dense interconnection subgraph, since the latter permits the existence 
of ``bottlenecks", whereas the requirement dictated by the axiom disallows 
such bottlenecks.

The third and the last property (iii)
%The forth property
%Axiom $\axiomConvergecast$ also concerns the boundary between the core and periphery, but 
%in addition it refers also to the structure of the periphery, and 
focuses on the flow of information from the periphery to the core
and measures its efficiency.
The core-periphery structure of the network is said to be
a {\em $\gamma$-convergecaster} if this data collection operation 
can be performed in time $\gamma$. 
The corresponding Axiom $\axiomConvergecast$ postulates that 
information can flow from the periphery nodes to the core efficiently 
(i.e., in constant time), namely, the core and periphery must form a $\Theta(1)$-convergecaster.
Note that one implication of this requirement is that the presence of periphery nodes that are far away from the core, or bottleneck edges that bridge between many periphery nodes and the core, is forbidden.

\begin{table}[t]
\centering
\begin{tabular}{|l|c|c|c|}\hline

Task & Running time & 
%Necessary Axioms & 
\multicolumn{2}{c|}{Lower bounds} \\ 
\cline{3-4}
& on $\C\P$ networks & All  Axioms & Any 2 Axioms\\
\hline\hline

%$k$-dissemination (for $k=1$ we get broadcast) & $O(k)$ & 
%Worst case graph: $\Omega(D+k/d_{\min})$, &
%
%$\C\P: \Omega(k/d_{\min})$&
%$\axiomClique$,$\axiomConvergecast$
%\\\hline\hline

MST  * & $O(\log^2n)$ &
%$\axiomBoundary$,$\axiomClique$,$\axiomConvergecast$ & 
$\Omega(1)$ & $\tilde{\Omega}(\sqrt[4]{n})$ \\ \hline

Matrix transposition & 
$O(k)$ & 
%$\axiomBoundary$,$\axiomClique$,$\axiomConvergecast$ &
$\Omega(k)$ & $\Omega(n)$ 
\\\hline

Vector by matrix multiplication & 
$O(k)$ & 
%$\axiomBoundary$,$\axiomClique$,$\axiomConvergecast$ &
$\Omega(k/\log n)$ & $\Omega(n/\log n)$
\\\hline

Matrix multiplication & $O(k^2)$ & 
%$\axiomBoundary$,$\axiomClique$,$\axiomConvergecast$ &
$\Omega(k^2)$ & $\Omega(n/\log n)$
\\\hline

Rank finding & $O(1)$ & 
%$\axiomBoundary$,$\axiomClique$,$\axiomConvergecast$ &
$\Omega(1)$ & $\Omega(n)$
\\\hline

Median finding & $O(1)$ & 
%$\axiomBoundary$,$\axiomClique$,$\axiomConvergecast$ &
$\Omega(1)$ & $\Omega(\log n)$
\\\hline

Mode finding & $O(1)$ & 
%$\axiomBoundary$,$\axiomClique$,$\axiomConvergecast$ &
$\Omega(1)$ & $\Omega(n/ \log n)$
 \\\hline

Number of distinct values & $O(1)$ &
%$\axiomBoundary$,$\axiomClique$,$\axiomConvergecast$ &
$\Omega(1)$ & $\Omega(n/\log n)$
\\\hline

Top $r$ ranked by areas & 
$O(r)$ & 
%$\axiomBoundary$,$\axiomClique$,$\axiomConvergecast$ &
$\Omega(r)$ & $\Omega(r\sqrt{n})$
\\\hline

%$k$ largest values &
%% (for $k=1$ we get find max) & 
%$O(k)$ &
%$\axiomClique$,$\axiomConvergecast$ &
%$\Omega(k)$&
%$\Omega(n)$
%\\\hline
 \multicolumn{4}{l}{$k$ - maximum number of nonzero entries in a row or column. * - randomized algorithm}

%Find average value & $O(1)$ & $\Omega(1)$& $\axiomClique$,$\axiomConvergecast$	\\\hline
\end{tabular}
\caption{\label{tab:intro_algo_summary}
Summary of algorithms for core-periphery networks.}
\end{table}

One may raise the ``semi-philosophical" question whether 
the properties we defined should be referred to as ``axioms". Our answer is that adopting the axiomatic view yields the added benefit that it immediately raises the fundamental issues of minimality, independence and necessity, thus allowing us to carefully verify the role and usefulness of each property / axiom.
%
%One may raise the question whether our 
%four 
%axioms are all essential.
Indeed, we partially address these issues.
% in two ways. 
First, we establish the independence of our axioms, by showing that
neither of them is implied by the other two.
Second, for each task, we establish the necessity of the axioms. Specifically, we show that if at least one of the axioms required by the algorithm is omitted, then there exists a network that satisfies the other axioms, but for which the running time (of {\em any} distributed algorithm) is larger by at least a factor of $\log n$ and at most a factor $n$, see Table 1.

To support and justify our selection of axioms, we examine their usefulness
% in the sense of their relevance to
for effective distributed computations on core-periphery networks. We consider a collection of different types of tasks, and show that they can be efficiently solved on core-periphery networks, by providing a distributed algorithm for each task and bounding its running time. 
%\michael{The following sentence says exactly the same as the last sentence of the previous paragraph. So, I suggest to remove the following sentence.}
%Moreover, for each task we argue the necessity of all three axioms, by showing that if at least one of the axioms is not satisfied by the network under consideration, then the same efficiency 
%can not be guaranteed by \emph{any} algorithm for the given task.

Table \ref{tab:intro_algo_summary} provides a summary of the main tasks we studied, along with the upper and lower bounds on the running time when the network satisfies our axioms, and a worst case lower bound on the time required when at least one of the axioms is not satisfied.
For each task we provide an algorithm, and prove formally its running time and the necessity of the axioms. 
As it turns out, some of the necessity proofs make use of an interesting connection to known communication complexity results.

%\begin{thm_our_algo_intro}[Meta Theorem]
%%\label{thm:mst_runtime}
%%\noindent{\bf Theorem:}
%If network $G(V,E)$ satisfies Core-Periphery Axioms $\axiomBoundary$,$\axiomClique$,$\axiomConvergecast$, 
%then there exists a distributed algorithm $\mathcal{A}$ for task  $\mathcal{B}$
%that run in $O(x)$ rounds with high probability \footnote{Hereafter, ``with high probability (w.h.p.)'' means 
%with failure probability polynomially small in $n$.}.
%%probability\footnote{With high probability (w.h.p.) means 
%%with probability at least $1-\tfrac{1}{n}$.}.
%More over if at least one axiom is not satisfied, there exist a graph $G(V,E')$ for which any algorithm 
%will need $w(x)$ rounds w.h.p.
%\end{thm_our_algo_intro} 

The most \emph{technically} challenging part of the paper is the distributed construction of a \emph{minimum-weight spanning tree} (MST), a significant task 
in both the distributed systems world, cf. \cite{lynch-book,peleg-book}, and the social networks world 
\cite{adamic1999small,bonanno2003topology,chen2003visualizing}.
Thus, the main \emph{algorithmic} result of the current paper is proving that MST can  
be computed efficiently (in $O(\log^2n)$ rounds) on core-periphery 
networks, namely, networks that comply with our axioms (interestingly, our algorithm is 
randomized, which is a rarity in the distributed MST literature).
To position this result in context, let us briefly review 
the state of the art on the problem of distributed MST construction.
The problem was first studied in \cite{GHS,awerbuch}, where the main focus 
was on low communication costs, and the run-time was at least linear in $n$. 
The study of sublinear-time distributed MST construction was initiated 
in~\cite{garay1998sublinear}. 
Currently, the best upper bound for general graphs is 
$O\left(\sqrt{n} \log^*n+ D\right)$ \cite{kutten1998fast},
where $D$ denotes the network diameter. %\footnote{maximal inter-node distance}. 
Conversely, it was shown in \cite{peleg2000near} that there is a graph for which any deterministic algorithm for MST requires $\Omega(\sqrt{n}/\log n)$ time. 
%$\Omega\left(\frac{\sqrt{n}}{\log n}\right)$. 
More efficient algorithms exist for specific families of graphs.
For the complete graph $G=K_n$, an MST, can be constructed in a distributed manner in $O(\log \log n)$ time \cite{lotker2005minimum}. 
For the wider class of graphs, of diameter at most 2, this task can still 
be performed in time $O(\log n)$. In contrast, taking the next step, 
and considering graphs of diameter 3, drastically changes the picture, 
as there are examples of such graphs for which any distributed MST construction 
requires $\Omega\left( \sqrt[4]{n} \right)$ time \cite{lotker06distributed}. 

Let us now briefly review related work.
Core periphery structures in social networks provided the inspiration to our approach. A number of excellent books provide a general review of social networks. 
A brief description of the core-periphery structure, and the core's possible 
utilization as an interconnection mechanism among the network's users 
(albeit with no formal models or algorithmic procedures), 
can be found in Easley and Kleinberg's book \cite{kleinberg-book}.
%page 552.
%although the book mentions the description of core-periphery,
%and says that people find other people using the core
%there are no models in the book.
The first explicit treatment of the core-periphery structure in social networks
is given by Borgatti and Everett in \cite{borgatti2000models}, 
which provides a descriptive model for the core-periphery structure, 
and reviews some of its occurrences in a variety of social settings, 
lending support to our intuition regarding the centrality of 
the core-periphery structure in social networks.
%they say:
%"The intuitive conception entails a dense, cohesive core and a sparse, unconnected periphery".
%We paraphrase this sentence earlier in our introduction.
It does not, however, provide any systematic mathematical model 
for this structure.
Similar phenomena of core-periphery structure were shown to exist in economics, e.g., the core--periphery model of Krugman \cite{krugman-1991} and other network formation models \cite{hojman2008core}.
Several books present a formal rigorous analysis for the Preferential 
Attachment model. The interested reader is invited to look at, cf., 
Chapter 3 of \cite{chung2006complex}, Chapter 4 of \cite{bonato2008course}, 
or Chapter 14 of  \cite{newman2010networks}. 

Turning to the distributed realm, in recent years there have been many studies focusing on the analysis of distributed algorithms on the complete graph $K_n$, cf. \cite{patt2011round,berns2012super,dolev2012tri,jung2012distributed,elkin2006unconditional,lotker06distributed,peleg2000near,lotker2005minimum,Lenzen:2011}. We believe that some of those algorithms can be extended to, and applied in, the context of core-periphery networks, which provide a much larger family of graphs.

The rest of the paper is organized as follows. Section \ref{sec:axioms} formally describes core-periphery networks, the axioms and their 
basic structural implications. Section \ref{sec:mst} provides a description of the MST algorithm, and Section \ref{sec:more_algorithms} presents the rest of the tasks we study.

\section{Axiomatic design for core-periphery networks}\label{sec:axioms}
% !TEX root =  core_per.tex

\subsection{Preliminaries}

Let $G(V,E)$ denote our (simple undirected) network, where $V$ is the set 
of nodes, $\abs{V}=n$, and $E$ is the set of edges, $\abs{E}=m$.
The network can be thought of as representing a distributed system.
We assume the synchronous {\cal CONGEST} model (cf. \cite{peleg-book}), 
where communication proceeds in \emph{rounds}, and in each round each node 
can send a message of at most $O(\log n)$ bits to each of its neighbors.
Initially, each node has a unique ID of $O(\log n)$ bits. 

For a node $v$, let $N(v)$ denote its set of neighbors and $d(v)=\abs{N(v)}$ 
its degree. 
For a set $S \subset V$ and a node $v \in S$, let 
$N_{\mathrm{in}}(v, S) = N(v) \cap S$ denote its set of neighbors within $S$, 
and denote the number of neighbors of $v$ in the set $S$ by
$d_{\mathrm{in}}(v,S)=\abs{N_{\mathrm{in}}(v, S)}$.
Analogously, let $N_{\mathrm{out}}(v,S) = N(v) \cap V \setminus S$ denote $v$'s set of neighbors outside the set $S$, and let $d_{\mathrm{out}}(v)=\abs{N_{\mathrm{out}}(v,S)}$.
For two subsets $S, T \subseteq V$, % s.t. $T \cap S = \emptyset$
let $\partial(S,T)$ be the \emph{edge boundary} (or cut) of $S$ and $T$, 
namely, the set of edges with exactly one endpoint in $S$, one in $T$ 
and $\abs{\partial(S,T)} =  \sum_{v \in S} \abs{N_{\mathrm{out}}(v,S) \cap T}$.
Let $\partial(S)$ 
denote the boundary in the special case where $T=V\setminus S$.

\subsection{Core-periphery networks}

Given a network $G(V,E)$, a $\cp$-partition 
is a partition of the nodes of $V$ into two sets, 
the \emph{core} $\C$ and the \emph{periphery} $\P$. 
Denote the sizes of the core and the periphery by $\nc$ and $\np$, respectively.
To represent the partition along with the network itself, we denote the 
{\em partitioned network} by $\GVECP$.

Intuitively, the core $\C$ consists of a relatively small group 
of strong and highly connected machines, designed to act as central servers, 
whereas the periphery $\P$ consists of the remaining nodes, typically acting
as clients. The periphery machines are expected to be weaker and less 
well connected than the core machines, and they perform much of their 
communication via the dense interconnection network of the core.
In particular, a central component in many of our algorithms, 
for various coordination and computational tasks, is based on assigning 
each node $v$ a {\em representative} core node $r(v)$, which is
essentially a neighbor acting as a ``channel" between $v$ and the core. 
The representative chosen for each periphery node is fixed.

For a partitioned network to be effective, the $\cp$-partition must possess
certain desirable properties. In particular, a partitioned network $\GVECP$ 
is called a {\em core-periphery network}, or {\em $\C\P$-network} for short, 
if the $\cp$-partition satisfies three properties, defined formally later on, 
in the form of three axioms.

\subsection{Core-periphery properties and axioms}

We define certain key parametrized properties 
of node groups, in networks that are of particular relevance to the 
relationships between core and periphery in our partitioned 
network architectures.
We then state our axioms, which capture the expected behavior 
of those properties in core-periphery networks, and demonstrate 
their independence and necessity. Our three basic properties are the following.

\begin{definition_bal}%[$\alpha$-Balanced Boundary]
A subset of nodes $S$ is said to have an {\em $\alpha$-balanced boundary} iff
$\frac{d_{\mathrm{out}} (v, S)}{d_{\mathrm{in}} (v, S)+1} =   O(\alpha)$
for every node $v \in S$. 
\end{definition_bal}
\begin{definition_cli}%[$\beta$-Clique Emulation]
The task of {\em clique emulation} on an $n$-node graph $G$ involves delivering 
a distinct message $M_{v,w}$, from $v$ to $w$, for every pair of nodes 
$v,w$ in $V(G)$. An $n$-node graph $G$ is a {\em $\beta$-clique-emulator},
if it is possible to perform clique emulation on $G$ within $\beta$ rounds 
(in the CONGEST model).
\end{definition_cli}
\begin{definition_con}%[$\gamma$-convergecast] 
For $S,T\subseteq V$, 
the task of {\em $\langle S,T\rangle$-convergecast}, on a graph $G$, involves
delivering $\abs{S}$ distinct messages $M_v$, originating at the nodes $v\in S$,
to some nodes in $T$ (i.e., each message must reach at least one node in $T$).
The sets $S,T\subset V$ form a {\em $\gamma$-convergecaster} 
if it is possible to perform $\langle S,T\rangle$-convergecast on $G$ 
in $\gamma$ rounds (in the CONGEST model).
\end{definition_con}

Consider a partitioned network $\GVECP$.
We propose the following set of \emph{axioms}, concerning 
the core $\C$ and periphery $\P$.

\begin{itemize}
[labelindent=1.5em,labelsep=*,leftmargin=*,itemsep=4pt]

\item[$\axiomBoundary$.] \textbf{Core Boundary.} 
%\\ {\bf DP: Perhaps ``Bounded Influence Ratio''} \\
The core $\C$ has a $\Theta(1)$-balanced boundary.

\item[$\axiomClique$.] \textbf{Clique Emulation.}
The core $\C$ is a $\Theta(1)$-clique emulator.

\item[$\axiomConvergecast$.] \textbf{Periphery-Core Convergecast}.
The periphery $\P$ and the core $\C$  form a $\Theta(1)$-convergecaster.
\end{itemize}

Let us briefly explain the axioms.    
Axiom $\axiomBoundary$ talks about the boundary between the core and periphery.
Core nodes with a high out-degree (i.e., with many links to the 
periphery) are thought of as ambassadors of the core to the periphery. 
Axiom $\axiomBoundary$ states that while not all nodes in the core must serve 
as ambassadors, if a node is indeed an ambassador, then it must also have many 
links {\em within} the core. 
Axiom $\axiomClique$ talks about the flow of information within the core, and postulates that the core must be dense, and in a sense, behave almost like a complete graph: ``everyone must know almost everyone else". The clique-emulation requirement is actually stronger than just being a dense subgraph, since the latter permits the existence of ``bottlenecks" nodes, which a clique-emulator must avoid.
Axiom $\axiomConvergecast$ also concerns the boundary between the core and periphery, but in addition it refers also to the structure of the periphery. It postulates that information can flow efficiently from the periphery to the core. For example, it forbids the presence of periphery nodes that are far away from the core, or bottleneck edges that bridge between many periphery nodes and the core.
Fig. \ref{fig:examples}(I) 
provides an example for a $\C\P$-network satisfying the three axioms. 

We next show that the axioms are independent.
Later, we prove the necessity of the axioms for the efficient performance 
of a variety of computational tasks.

\begin{theorem}
\label{thm:independent}
Axioms $\axiomBoundary$, $\axiomClique$, $\axiomConvergecast$ are independent, 
namely, assuming any two of them does not imply the third.
\end{theorem}

\begin{figure}[t]
\centering
\begin{tabular}{cc}%35,45
\includegraphics[width=.38\columnwidth]{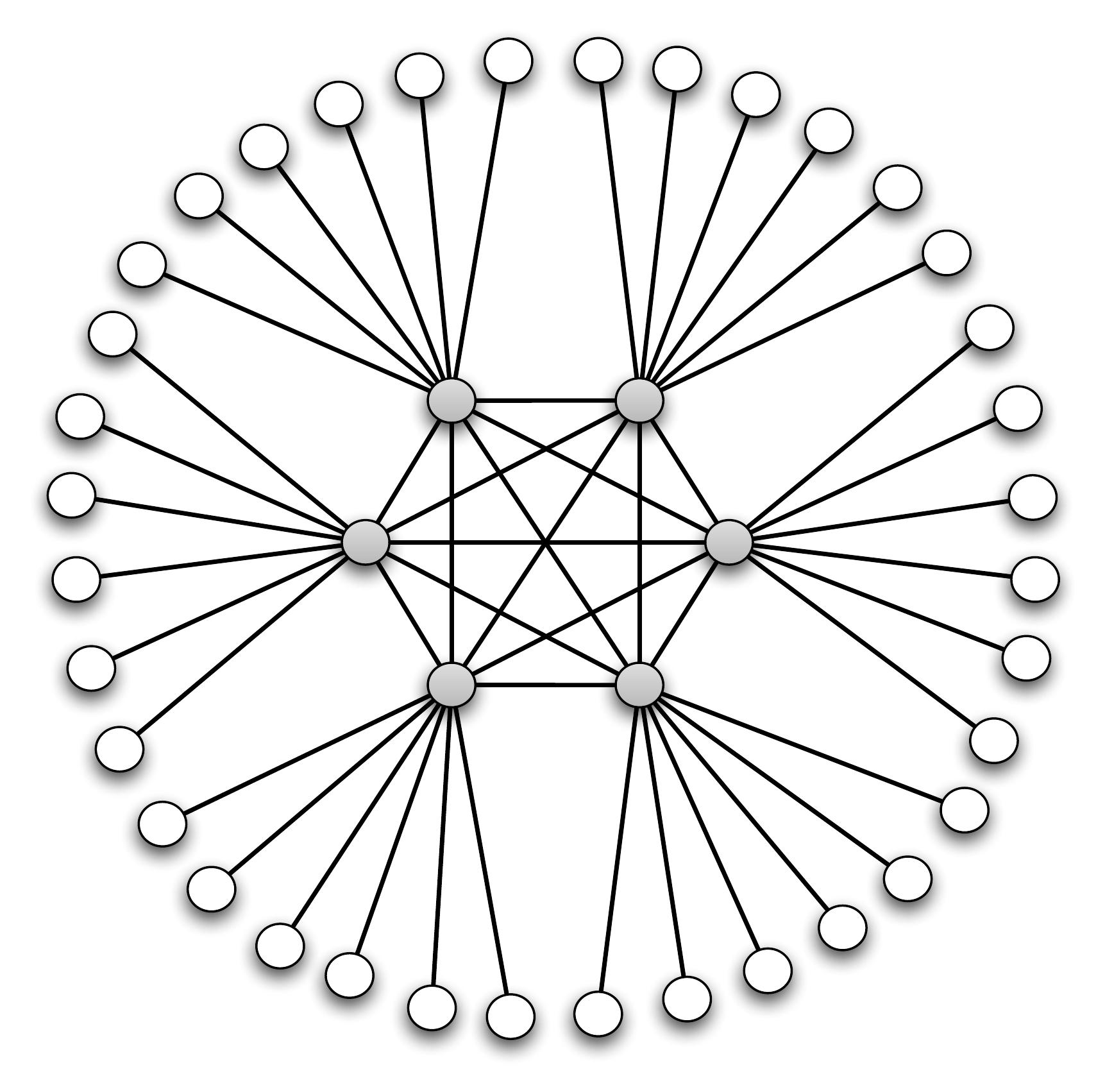} &
\includegraphics[width=.48\columnwidth]{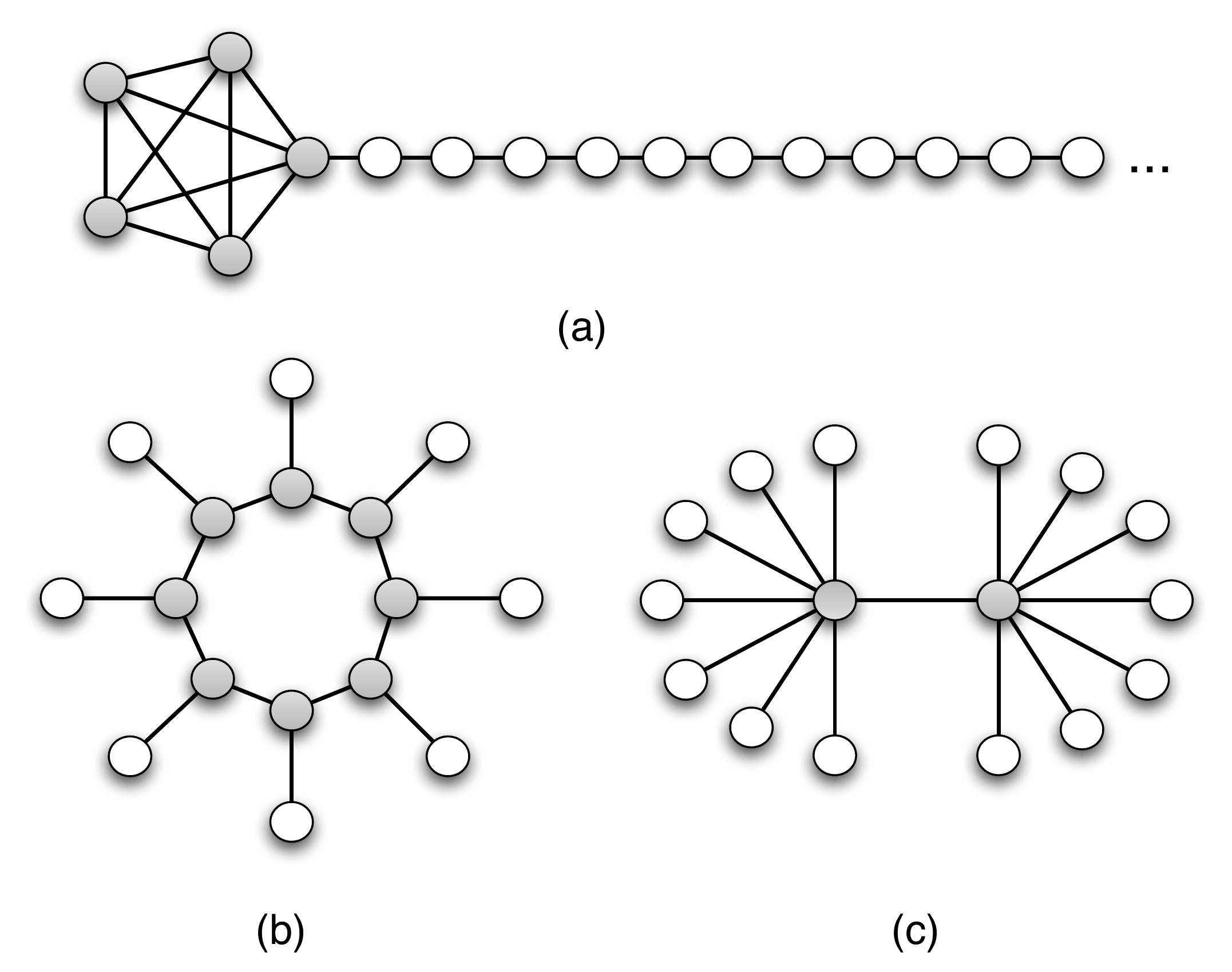} \\
(I)  &   (II)
\end{tabular}
\caption{
(I) An example for a 36-node $\C\P$-network that satisfies all three axioms. The 6 core nodes (in gray) are connected in clique.  In this example every core node is also an ambassador with equal number of edges to the core and outside the core. The core and periphery form a convergecaster since the periphery can send all its information to the core in one round.
(II) Networks used in proofs: (a) The ``lollipop partitioned" network $L_{25}$. (b) The ``sun partitioned" network $S_{16}$. (c) The ``dumbbell partitioned" network $D_{16}$ .}
\label{fig:examples}
\end{figure}

\begin{proof}
We prove the theorem by considering three examples of partitioned networks, described next. Each of these networks satisfies two of the axioms, but violates the third
(hence, they are not $\C\P$-networks), implying independence.

\paragraph{The lollipop partitioned network $L_n$ (Fig. \ref{fig:examples}(II)(a))}

The lollipop graph consists of a $\lfloor\sqrt{n}\rfloor$-node clique and a 
$n-\lfloor\sqrt{n}\rfloor$-node line, attached to some node of the clique. 
The corresponding partitioned network is obtained by setting its core $\C$
to be the clique, and its periphery $\P$ to be the line. 
Observe that $L_n$ is not a $\C\P$-network. Indeed, Axiom $\axiomClique$ holds on $L_n$, 
and Axiom $\axiomBoundary$ also holds since 
the outgoing degree of each node in the core is 0 or 1.
%except one node which has outgoing degree $1$. 
However, $\axiomConvergecast$ is not satisfied on $L_n$
since the periphery consists of  a line of length $n-\lfloor\sqrt{n}\rfloor$, 
so it will take linear time for the periphery $\P$ 
to convergecast to the core $\C$.

\paragraph{The sun partitioned network $S_n$ (Fig. \ref{fig:examples}(II)(b)) } 

The sun graph consists of an $\lceil n/2\rceil$-node cycle, with an additional leaf node
attached to each cycle node. 
The corresponding partitioned network is obtained by setting its core $\C$
to be the cycle, and its periphery $\P$ to contain all other $\lfloor n/2\rfloor$ nodes. 
Clearly, $\axiomConvergecast$ holds on $S_n$, since each node in $\P$ is only 
one hop away from some node in $\C$. Axiom $\axiomBoundary$ also holds, 
since the outgoing degree of each node in $\C$ is $1$. 
Axiom $\axiomClique$, however, does not hold, since the distance 
between two diametrically opposing nodes in the cycle is at least $n/4$, 
so it is not possible to perform $\Theta(1)$ clique emulation.

\paragraph{The dumbbell partitioned network $D_n$ (Fig. \ref{fig:examples}(II)(c)) } 
The dumbbell graph is composed of two stars, 
each consisting of a center node connected to $\lceil n/2\rceil-1$ leaves,
whose centers are connected by an edge. 
The corresponding partitioned network is obtained by setting its core $\C$
to be the two centers, and the periphery $\P$ to consist of the $n-2$ leaves 
of the two stars (each of degree $1$). 
It is easy to see that Axioms $\axiomClique$ and $\axiomConvergecast$ hold on $D_n$,
while Axiom $\axiomBoundary$ does not. 
\end{proof}

\subsection{Structural implications of the axioms}

The axioms imply a number of simple properties of the network structure.

\begin{theorem}
\label{thm:cp_property2}
If the partitioned network $\GVECP$ is a core-periphery network (i.e., it 
satisfies Axioms $\axiomBoundary$, $\axiomClique$ and $\axiomConvergecast$), 
then the following properties hold:
\begin{enumerate}
\item The core size satisfies $\Omega(\sqrt{n}) \le \nc\le O(\sqrt{m})$. 
%\Omega(\sqrt{n})$ and $\nc = O(\sqrt{m})$.
\item Every $v\in\C$ satisfies 
$d_{\mathrm{out}} (v,\C) = O(\nc)$ and $d_{\mathrm{in}}(v, \C)= \Omega(\nc)$.
\item The number of outgoing edges of the core is $\abs{\partial(\C)} = \Theta(\nc^2)$.
\item The core is dense, i.e., the number of edges in it is 
$\sum_{v \in \C} d_{\mathrm{in}} (v,\C) = \Theta(\nc^2)$.
\end{enumerate}
\end{theorem}
\begin{proof}
%Let $\nc$ be the core size. 
Axiom $\axiomClique$ necessitates that the inner degree 
of each node $v$ is $d_{\mathrm{in}} (v,\C) = \Omega(\nc)$ (or else 
it would not be possible to complete clique emulation in constant time),
implying the second part of claim 2. 
It follows that the number of edges in the core is $\sum_{v \in \C} d_{\mathrm{in}} (v,\C) = \Theta(\nc^2)$, hence it is dense; claim 4 follows.
Since also $\sum_{v \in \C} d_{\mathrm{in}} (v,\C) \le 2m$, we must have the upper bound of claim 1, that is, $\nc = O(\sqrt{m})$.
Axiom $\axiomBoundary$ yields that for every $v$,  $d_{\mathrm{out}} (v,\C) =  O(\nc)$, so the first part of claim 2 follows.
Note that $\abs{\partial(\C)} = \sum_{v \in \C} d_{\mathrm{out}} (v,\C) = O(\nc^2)$,
so the upper bound of claim 3 follows. 
To give a lower bound on $\nc$, note that by Axiom $\axiomConvergecast$ we have
$\abs{\partial (\C)} = \Omega(n-\nc)$ (otherwise the information from the $n-\nc$ nodes of $\P$ could not flow in $O(1)$ time to $\C$), so $\nc = \Omega(\sqrt{n})$ and the lower bounds of claims 1 and 3 follow.
\end{proof}

An interesting case for efficient networks is where the number of edges is linear in the number of nodes. In this case, Theorem  \ref{thm:cp_property2} implies the following.

\begin{coro}
In a core-periphery network $\GVECP$ where $m=O(n)$, the following properties hold:
\begin{enumerate}
\item The core size satisfies $\nc =\Theta(\sqrt{n})$.
\item The number of outgoing edges from the core is $\abs{\partial(\C)} = \Theta(n)$.
\item The number of edges in the core is  
$\sum_{v \in \C} d_{\mathrm{in}} (v,\C) = \Theta(n)$.
\end{enumerate}
\end{coro} 

Now we show a key property relating our axioms to the network diameter.

\begin{clm}
\label{clm:1}
If the partitioned network $\GVECP$ 
satisfies Axioms $\axiomClique$ and $\axiomConvergecast$, 
then its diameter is $\Theta(1)$.
\end{clm}
\begin{proof}
Suppose the partitioned network $\GVECP$ satisfies Axioms 
$\axiomClique$ and $\axiomConvergecast$. 
Let $u,v$ be two nodes in $V$. There are three cases to consider: 
(1) Both $u,v \in \C$: then Axiom $\axiomClique$ ensures $O(1)$ time message 
delivery, and thus $O(1)$ distance. 
(2) $u\in\P$ and $v\in \C$: Axiom $\axiomConvergecast$ implies that 
there must be a node in $w\in\C$ such that $dist(u,w)=O(1)$. 
By Axiom $\axiomClique$, and $dist(w,v)=O(1)$, thus $dist(u,v)\le dist(u,w)+dist(w,v)=O(1)$.
(3) Both $u,v\in P$: then there must be $w,x\in\C$ 
such that $dist(u,w)=O(1)$ and $dist(x,v)=O(1)$. Since $dist(w,x)=O(1)$ 
by Axiom $\axiomClique$, it follows that $dist(u,v)\le dist(u,w)+dist(w,x)+dist(x,v)=O(1)$.
Hence $dist(u,v)=O(1)$ for any $u,v\in V$, so the diameter of $G(V,E)$ is constant.
\end{proof}

\noindent The following claim shows that the above conditions are necessary.

\begin{clm}
\label{clm:exist_graph_diam_n} 
For $X\in\{E,C\}$, there exists a family of $n$-node partitioned networks 
$G_X(V,E,\C,\P)$, of diameter $\Omega(n)$, 
which satisfy all axioms except $\mathcal{A}_{X}$.
\end{clm}
\begin{proof}
For $X=C$, let $G_C(V,E,\C,\P)$ be the lollipop partitioned network $L_n$.
As mentioned before, for this network Axiom $\axiomConvergecast$ is violated, 
while the others are not. Also note that the diameter of $G_C$ is $\Omega(n)$.

For $X=E$, let $G_E(V,E,\C,\P)$ be the sun partitioned network $S_n$.
As mentioned before, for this network Axiom $\axiomClique$ is violated, 
while the others are not. Also note that the diameter of $G_E$ is $\Omega(n)$.
\end{proof}

\section{MST on a Core-Periphery Network}\label{sec:mst}
% !TEX root = core_per.tex

In this section we present a time-efficient randomized distributed 
algorithm $\C\P$-MST for computing a \emph{minimum-weight spanning tree} (MST) 
on a core-periphery network.
% The output should be distributed, namely, 
In particular, we consider a $n$-node core periphery network $\GVECP$,
% with a core $\C$ and a periphery $\P$ 
namely, a partitioned network satisfying all three axioms, 
and show that a MST can be computed in a distributed manner on such a network 
in $O(\log^2 n)$ rounds with high probability. 
%(i.e., with failure probability polynomially small in $n$).
Upon termination, each node knows which of its edges belong to the MST. 
We also show that Axioms $\axiomBoundary$, $\axiomClique$, and 
$\axiomConvergecast$ are indeed necessary, for our 
distributed MST algorithm to be efficient. 

%%%%%%%%%%%%%%%%%%%%%%%%%%%%
\subsection{Axiom necessity}
%%%%%%%%%%%%%%%%%%%%%%%%%%%%

\begin{theorem}\label{thm:MST-necessity}
For each $X \in\{B,E,C\}$ there exists a family of $n$-node partitioned networks
${\cal F}_X = \{ G_X(V,E,\C,\P)(n)\}$,
that do not satisfy Axiom $\mathcal{A}_X$, but satisfy the other two axioms; 
and the time complexity of \emph{any} distributed MST algorithm on ${\cal F}_X$ is $\Omega(n^{\alpha_X})$, for some constant $\alpha_X>0$.
\end{theorem}

\begin{proof}
For $X=B$, consider the graph $G_B$ on Figure \ref{fig:g1g2}(a), in which Core is a clique of size $k$, and each node in the Core is connected to $k^3$ Periphery nodes (one node in Core is also connected to $s$, so it has $k^3 +1$ Periphery neighbors). The number of nodes in $G_B$ is thus $n=k+k\cdot k^3 + 1=\Theta(k^4)$. 
In \cite{lotker06distributed}, it was shown that any distributed algorithm will take at least $\Omega \left(\sqrt[4]{n}/\sqrt{\log n}\right)$ time on $G_B$. Since Core is a clique, $G_B$ satisfies Axiom $\axiomClique$. Since every node in Periphery has a direct edge to the Core, $G_B$ satisfies Axiom $\axiomConvergecast$, i.e., it is possible to perform a convergecast in $O(1)$ time. But notice that $d_{\mathrm{in}}=\sqrt[4]{n}$ while $d_{\mathrm{out}}=\sqrt[4]{n^3}$ and thus $G_B$ does not satisfy Axiom $\axiomBoundary$.

%\item 
For $X=E$, consider the graph $G_E$ on Figure \ref{fig:g3g3mst}(a), in which Core is a collection of $k$ cliques, each of size $k$, where a single node in each clique is connected to a special Core node $u$, and there are no edges between cliques. 
The $k^3$ Periphery nodes are arranged in $k$ columns of $k^2$ nodes each.
Each node in the Core (except $u$) is connected to $k$ Periphery nodes such that the nodes in a specific clique $i$ are connected to all the Periphery nodes that reside in a specific column $i$. One Core node (from the leftmost clique) is additionally connected to $s$, and another Core node (from the rightmost clique) is connected to $r$. The number of nodes in $G_E$ is thus $n=k\cdot k+k\cdot k^2 + 2=\Theta(k^3)$. 

%Let's take a look at the nodes $s$ and $r$, and 
Assume the following weight assignment. All the edges between Core and Periphery have weight $10$, except for the two edges that come from $s$ and $r$. The weights of all the edges incident to $s$ are $2$, and the weights of all the edges incident to $r$ are $3$. Assume also that the weights of all the rest of the edges in Periphery are $1$. It's easy to see that such a weight assignment will yield an MST as illustrated in Figure \ref{fig:g3g3mst}(b). Notice that increasing the weight of some edge incident to $s$ (say, to $5$), will cause this edge to be removed from the MST, and the corresponding edge incident to $r$ to be included. Thus, in order for $r$ to know which of its edges belong to the MST, it needs to receive information regarding the weights of all the edges incident to $s$, i.e., at least $k^2$ edge weights should be delivered from $s$ to $r$. 
In the CONGEST model, at most $O(\log n)$ edge weights can be sent in a single message, hence $\Omega(k^2/\log n)$ messages must be delivered from $s$ to $r$.
Next, we show that delivering $k^2/\log n$ messages from $s$ to $r$ will require at least $\Omega(k/(\log n))$ time. First, note that any path $s\rightarrow r$ that is not passing via the node $u$ has length at least $k$, thus if any of the messages avoids $u$, we are done. So now assume that all the messages take paths via $u$. Observe that the edge cut of the node $u$ (i.e., its degree) is $k$, and thus in $k/(2\log n)$ time units, it can forward at most $k^2/(2\log n)$ messages, which is not sufficient for completing the MST task. Thus, any MST algorithm on the graph $G_E$ will take at least $\Omega(k/\log n)=\Omega(\sqrt[3]{n}/\log n)$ time.

It is left to show that $G_E$ satisfies Axioms $\axiomBoundary$ and $\axiomConvergecast$, but not $\axiomClique$. For every node in the Core, $d_{\mathrm{in}}=k$ and $d_{\mathrm{out}}=k$, except the node $u$, for which $d_{\mathrm{out}}=0$. So, for each node in the core $d_{\mathrm{out}}/(d_{\mathrm{in}}+1)=O(1)$, which means that $\axiomBoundary$ is satisfied. Since every node in Periphery has a direct edge to the Core, $G_E$ satisfies Axiom $\axiomConvergecast$, i.e., it is possible to perform a convergecast in $O(1)$. It is also easy to see that the Core does not support $O(1)$-clique emulation ($\axiomClique$), since sending $k$ messages out of any clique in Core to any other clique in Core requires $k$ time, as there is only one edge connecting any clique to the node $u$.

%\item 
Finally, for $X=C$, consider a graph $G_C$ on Figure \ref{fig:g1g2}(b), in which Core is a clique of size $k$, and each node in the Core is connected to $k/2$ Periphery nodes. One Core node is additionally connected to a cycle of size $k^2/2$ that resides in Periphery. The number of nodes in $G_C$ is thus $n=k+k\cdot k/2 + k^2/2=\Theta(k^2)$. 
It is easy to see that Axioms $\axiomBoundary$ and $\axiomClique$ are satisfied, but Axiom $\axiomConvergecast$ is not. For a suitable weight assignment, the decision regarding which edge of $r$ to include in the MST depends on the weights of the edges incident to $s$. The last observation implies that at least one message has to be delivered from $s$ to $r$ which will take $\Omega(k^2)=\Omega(n)$ time.
\end{proof}

\begin{figure}[!ht]
\centering
\includegraphics[width=.35\columnwidth]{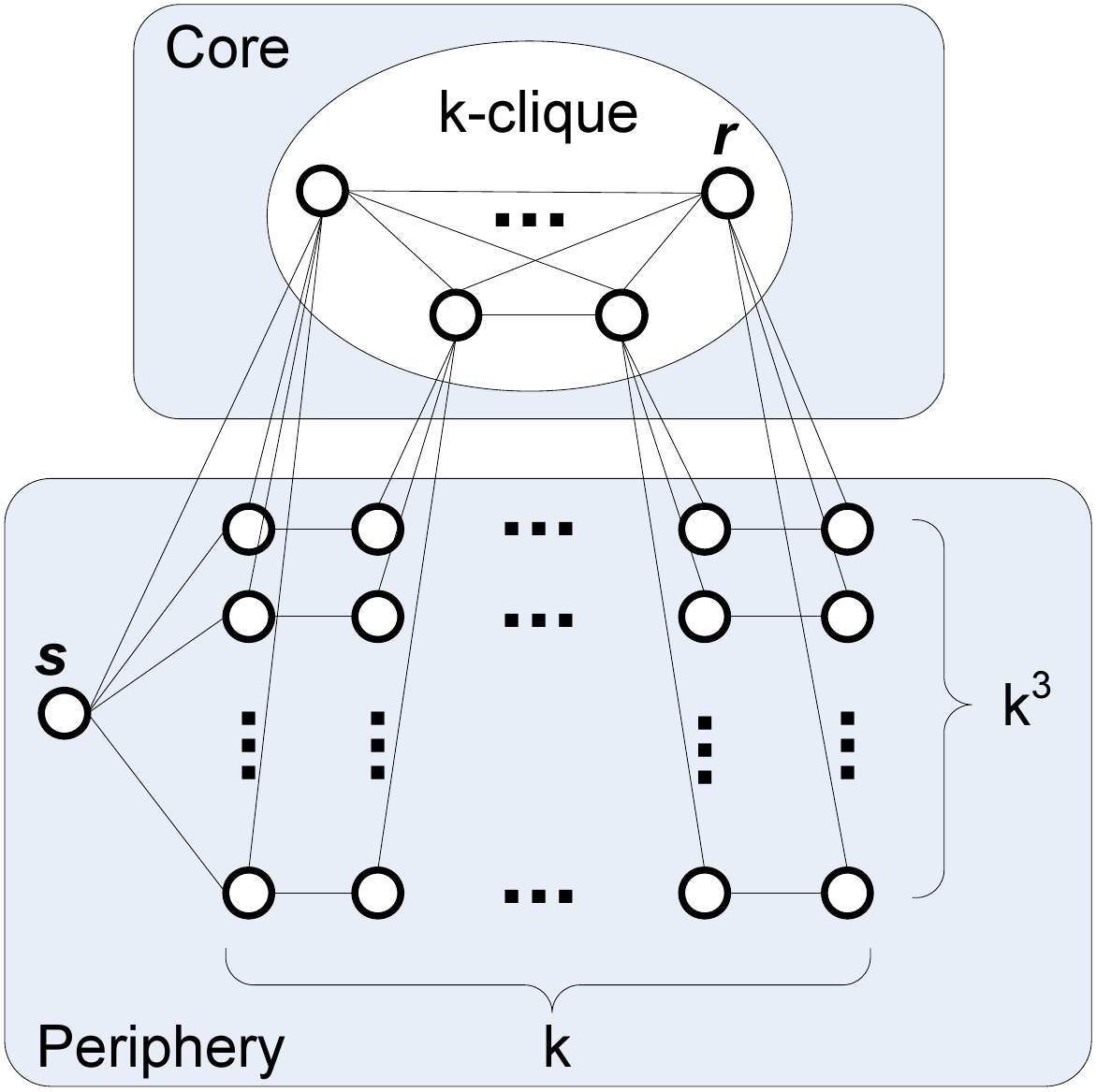}~~~~~~~
\includegraphics[width=.35\columnwidth]{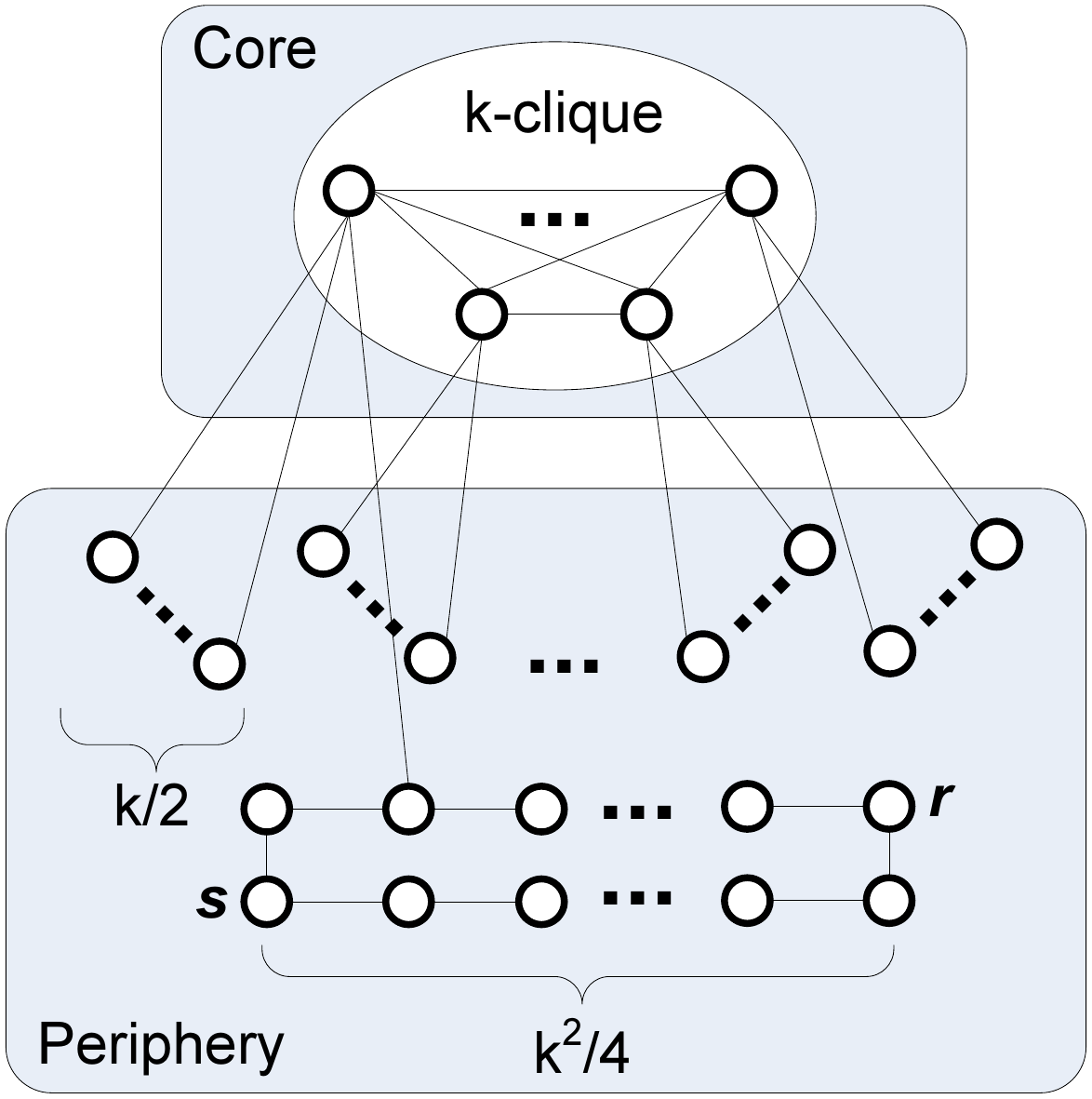} \\\vspace{1mm}
{\small(a) Graph $G_B$\hspace{0.25\columnwidth}(b) Graph $G_C$}\\
\caption{(a) Graph $G_B$: each node in the Core is connected to $k^3$ Periphery nodes. (b) Graph $G_C$: each node in the Core is connected to $k/2$ Periphery nodes, and one Core node is connected to cycle of length $k^2/2$ in Periphery.}\label{fig:g1g2}
\end{figure}

\begin{figure}[!ht]
\centering
\includegraphics[width=.49\columnwidth]{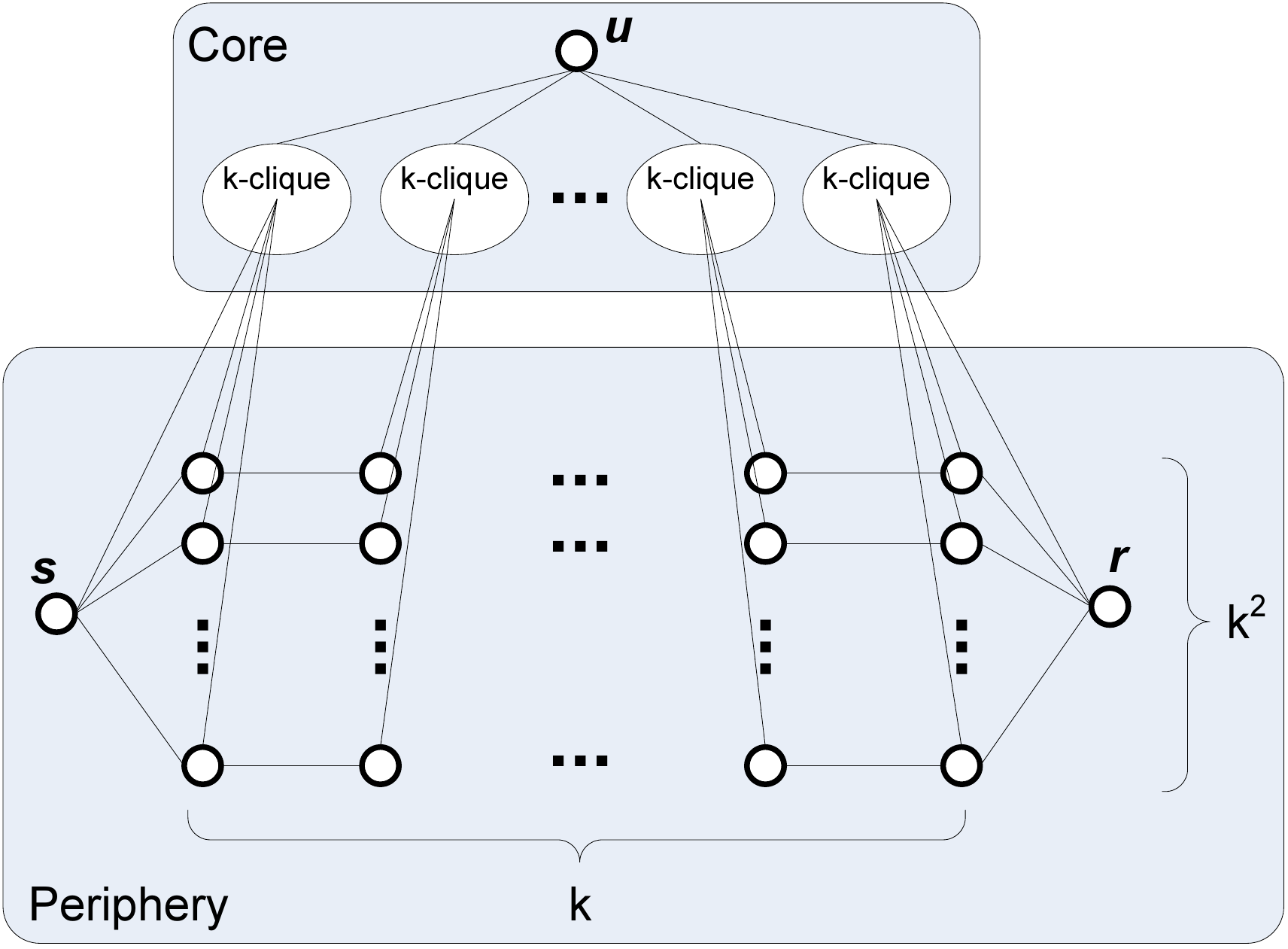}
\includegraphics[width=.49\columnwidth]{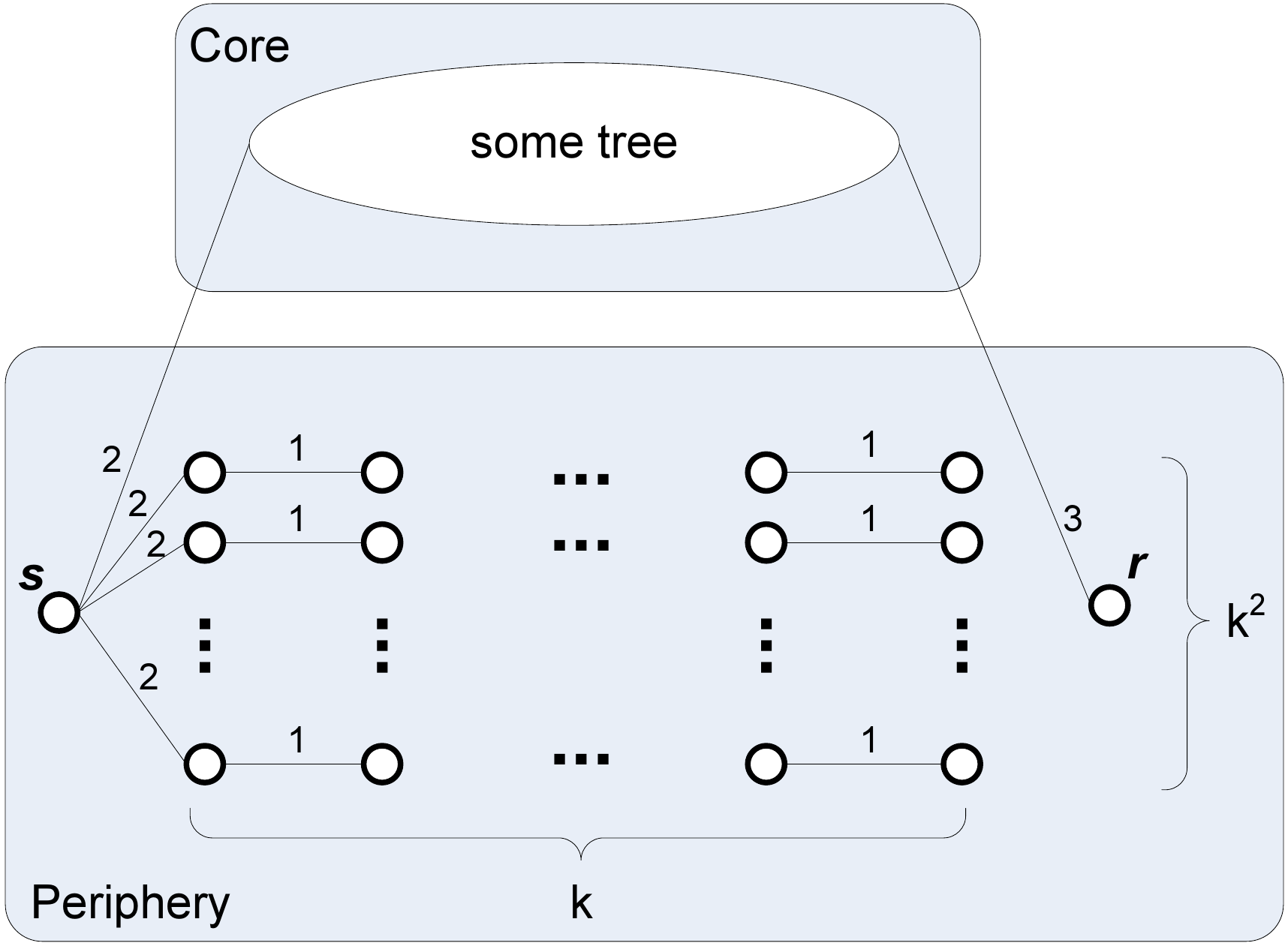} \\\vspace{1mm}
{\small(a) Graph $G_E$\hspace{0.35\columnwidth}(b) MST}\\
\caption{(a) Graph $G_E$: Core consists of $k$ cliques each of size $k$. Each node in the Core (except $u$) is connected to $k$ nodes in Periphery. (b) Possible MST of $G_E$.}\label{fig:g3g3mst}
\end{figure}

%%%%%%%%%%%%%%%%%%%%%%%%%%
\subsection{Description of the $\C\P$-MST algorithm}
%%%%%%%%%%%%%%%%%%%%%%%%%%

Let us now give a high level description of our $\C\P$-MST algorithm. 
The algorithm is based on Boruvka's MST algorithm \cite{Boruvka:1926}, and 
%similarly to that algorithm it 
runs in $O(\log n)$ phases, each consisting of several steps.
The algorithm proceeds by maintaining a forest of {\em tree fragments} 
(initially singletons) and gradually merging fragments, until the forest 
converges to a single tree. 
Throughout the execution, each node has two \emph{officials}, namely, 
core nodes that represent it. In particular, recall that
%at the beginning of the execution, 
each node $v$ is assigned a {\em representative} core neighbor $r(v)$, 
passing information between $v$ and the core. 
In addition, $v$ is also managed by the \emph{leader} $l(i)$ of its 
current fragment $i$. An important distinction between these two roles is that 
the representative of each node is fixed, while its fragment leader may change 
in each phase (as its fragment grows). 
%During the execution, the algorithm maintains the invariant that 
At the beginning of each phase, every node knows the IDs of its fragment 
and its leader. 
Then, every node considers all its outgoing edges (i.e., edges with the second endpoint belonging to another fragment), and finds its minimum weight outgoing edge. This information is delivered to the core by the means of the representative nodes, which receive the information, aggregate it (as much as possible) and forward it to the leaders of the appropriate fragments. The leaders decide on the fragment merging, and inform all the nodes about new fragments IDs. 

The main challenges in obtaining the proof were in bounding the running time, 
which required careful analysis. 
There are two major sources of potential delays in the algorithm.
The first involves sending information between officials (representatives 
to leaders and vice versa).
%, in steps 3 and 5 in the main algorithm
Note that there are only $O(\sqrt{m})$ officials, but they may need to send 
information about $m$ edges, which can lead to congestion. 
For example, if more than $\alpha \cdot \sqrt{m}$ messages need to be sent
to an official of degree $\sqrt{m}$, then this will take at least $\alpha$ 
rounds. We use randomization of leaders and the property of clique emulation to avoid 
this situation, and make sure that officials do not have to send, or receive, 
more than $O(\sqrt{m} \log m)$ messages in a phase.

The second source for delays is the fragment merging procedure. 
This further splits into two types of problems.
The first is that a chain of fragments that need to be merged could be long, 
and in the basic distributed Boruvka's algorithm will take long time 
(up to $n$) to resolve. This problem is overcome by using a modified 
pointer jumping technique, similar to \cite{lotker06distributed}. 
The second problem is that the number of fragments that need to be merged 
could be large, resulting in a large number of \emph{merging} messages 
that contain, for example, the new fragment ID. This problem is overcome 
by using randomization, and by reducing the number of messages needed 
for resolving a merge.
%as described in detail in the next subsection.

\def\AppendixPseudocodes{
%\subsection{Pseudocodes}
\renewcommand{\thealgorithm}{}
\renewcommand{\algorithmicrequire}{Executed every phase by every leader $w\in\C$, for each fragment $i\in F_{lead}(w)$.}
\begin{algorithm}[H]
\caption{$\MF(i)$}
\label{alg:merge_frags}	
    \begin{algorithmic}[1]
		\REQUIRE
		\item[]
    		\IF {$state(i) = active$}
    		\STATE $next \gets \mp(i)$
    		\STATE $state(i) \gets frozen$
    		\STATE $[isFound, rootFrag]\gets \textsc{FindRoot}(i)$
    		\IF {$isFound = \textsc{true}$}
    		\STATE $state(rootFrag) \gets root$
    		\ENDIF
    		\ENDIF
    		\item[]
    		\IF {$state(i)=frozen$}
    		\STATE $[isFinished, next] \gets \pj(i,next,2)$	
    		\STATE {$next(F_{lead}^{\mp(i)}(w)) \gets next$}
    		\IF {$isFinished = \textsc{true}$ }
    		\STATE {$state(F_{lead}^{\mp(i)}(w)) \gets waiting$}
    		\ENDIF
    		\ENDIF
    		\item[]
    		\IF {$state(i)=waiting$}
    		\STATE receive merge-requests
    		\STATE send merge-replies with $next$ (which now points to the $root$)
    		\STATE wait for FIN msg from $root$ with $newID$ and $newLead$
    		\IF {FIN msg received}
    		\STATE $newID(F_{lead}^{\mp(i)}(w)) \gets newID$
    		\STATE $state(F_{lead}^{\mp(i)}(w))\gets active$
    		\ENDIF
    		\ENDIF
    		\item[]
    		\IF {$state(i)=root$}
    		\STATE wait for incoming merge-requests
    		\STATE store the sources of the requests
    		\STATE reply on all requests with $\textsc{null}$
    		\IF {num of requests $= 0$ and size of merge-tree $\le 2^{2+phase}$}
    		\STATE $newID \gets$ random ID among all fragments in the merge-tree
    		\STATE $newLead \gets$ random node in $\C$
    		\STATE send FIN msg with $newID$ and $newLead$ to all the stored sources
    		\STATE $state(i)\gets active$
    		\ENDIF
    		\ENDIF
    \end{algorithmic} 
\end{algorithm}

\renewcommand{\algorithmicrequire}{Executed by each fragment $i$ in the $frozen$ state}
\renewcommand{\algorithmicensure}{\textbf{Input:} $next$ -- first fragment to try, $iter$ -- how many pointer-jumps to perform}
\begin{algorithm}[H]
    \caption{$\pj(i,next,iter)$ (pointer jumping)}
\label{alg:ptr_jump}
    \begin{algorithmic}[1]
    		\REQUIRE
    		\ENSURE
    		\renewcommand{\algorithmicensure}{\textbf{Output 1:} indication whether the $root$ was reached}
    		\ENSURE
    		\renewcommand{\algorithmicensure}{\textbf{Output 2:} pointer to the $root$ or to the next fragment in the chain}
    		\ENSURE
    		\item[]
   			\WHILE {$iter > 0$}
   			\IF {$i=\spk^{\mp(i)}(w)$}
    		\STATE send merge-request to $next$
    		\ENDIF
    		\item[]
    		\STATE receive merge-requests
    		\STATE send merge-replies with $next$
    		\item[]
    		\IF {$i=\spk^{\mp(i)}(w)$}
    		\STATE receive merge-reply with $next'$
    		\IF {$next' = \textsc{null}$}
    		\RETURN $[\textsc{true},next]$
    		\ENDIF 
    		\STATE $next \gets next'$
    		\STATE $iter \gets iter - 1$ 
    		\ENDIF
    		\ENDWHILE
    		\item[]
    		\RETURN $[\textsc{false},next]$
    \end{algorithmic}
\end{algorithm}

\renewcommand{\algorithmicrequire}{On \textbf{reply} reception:}
\renewcommand{\algorithmicensure}{\textbf{Input:}}
\begin{algorithm}[H]
    \caption{$\textsc{FindRoot}(i)$}
\label{alg:isroot}
    \begin{algorithmic}[1]
    		\IF {$i=\spk^{\mp(i)}(w)$}
    		\STATE send merge-request to $\mp(i)$
    		\ENDIF
    		\item[]
    		\STATE receive merge-requests
    		\STATE send merge-replies with $\mp(i)$
    		\item[]
    		\IF {$i=\spk^{\mp(i)}(w)$}
    		\STATE receive merge-reply with $next'$
    		\IF {$next'\in F_{lead}^{\mp(i)}(w)$}
    		\IF {$next' \le \mp(i)$}
    		\RETURN $[\textsc{true}, next']$
    		\ENDIF 
    		\ENDIF
    		\ENDIF
    		\item[]
    		\RETURN $[\textsc{false},\textsc{NULL}]$
    \end{algorithmic}
\end{algorithm}

}%\AppendixPseudocodes

Before describing the subsequent phases of the algorithm, 
a few definitions are in place.
%First recall that $\mathrm{id}(u)$ is the ID of node $u$ after renaming, $r(u)$ is the representative of $u$ in the core and for a fragment $i \in [1,\ldots,n]$, $l(i)$ is $i$'s leader in the core.
At any point throughout the execution, let $f(u)$ denote the fragment that $u$ belongs to. Dually, 
let $V^i$ denote the set of nodes in fragment $i$, and let $V^i(w)$ denote 
the subset of $V^i$ consisting of the nodes that are represented by $w$.
For a representative $w\in\C$, 
%$w$ in the core, 
let $F_{rep}(w)$ be the set of fragments 
that $w$ represents, namely, $F_{rep}(w) = \{i \mid V^i(w)\ne\emptyset\}$,
%Alternative definitions:
%$F_{rep}(w) = \{i\mid\exists u\in\F(i)~\mbox{ s.t. }~r(u)=w\}$,
%$F_{rep}(w) = \{i\mid\exists u~\mbox{ s.t. }~r(u)=w~\mbox{ and }~f(u)=i\}$,
%$i \in F_{rep}(w)$ if there exist a node $u$ for which $r(u)=w$ and $f(u)=i$, 
and let $F_{lead}(w)$ be the set of fragments that $w$ leads, namely, 
$F_{lead}(w) = \{i\mid l(i)= w\}$.
% $i \in F_{lead}(w)$ if $l(i) = w$.
For a set of nodes $S^i$ belonging to the same fragment $i$, 
an {\em outgoing edge} is one whose second endpoint belongs to a different 
fragment. Let $\mwoe(S^i)$ be the \emph{minimum weight outgoing edge} of $S^i$.
For a node $u$, a fragment $i$ and a representative $w$, we may occasionally 
refer to the fragment's $\mwoe$ as either $\mwoe(u)$, $\mwoe(V^i)$ or 
$\mwoe(V^i(w))$.
The \emph{merge-partner} of fragment $i$, denoted $\mp(i)$, is the fragment 
of the second endpoint of the edge $\mwoe(V^i)$. 
%\par
Define $F_{lead}^j(w)\subseteq F_{lead}(w)$ to be the set of fragments led 
by $w$ that attempt to merge with the fragment $j$, i.e., 
$F_{lead}^j(w)=\{i\mid i\in F_{lead}(w) ~and~ \mp(i)=j\}$. 
%Let us also 
Define a \emph{speaker} fragment $\spk^j(w)=\min F_{lead}^j(w)$,
that is responsible for sending merge-requests on behalf of all the fragments 
in $F_{lead}^j(w)$, and updating them upon the reception of merge-replies.

We now proceed with the description of the algorithm.

\paragraph{Phase 0 -- Initialization}
\begin{enumerate}
\item \textbf{Obtaining a Representative.} Each node $u \in V$ obtains a representative $r(u)\in \C$ in the core.
In particular, if $u\in \C$, it represents itself, i.e., 
$r(u)=u$. Each periphery node $u\in \P$ sends a ``representative-request''
message towards the core $\C$ with its ID. This is done in parallel,
using a $\gamma$-convergecast protocol on $\P$ and $\C$,
which ensures that each such message is received by some node in $\C$.
%These ``representative-requests" are flooded in the network (since the route from $u$ to some node in $\P$ is not known at this stage). 
Once a node $w \in \C$ receives such a message, it replies to $u$ on the same 
route, and $u$ sets $r(u)=w$. 

%(Michael: using $\gamma$-emulation.)
\item \textbf{Renaming.} Each node $u \in V$ receives a unique ID, $\mathrm{id}(u) \in [1,\ldots,n]$. This step can be performed
%in constant number of rounds 
in the following simple way: 
each node sends to its representative its ID, and each representative sends 
its own ID and the number of nodes it represents, to all core members. 
Now, every core member can sort the core IDs and 
reserve a sufficiently large range of IDs
%leave enough room of IDs 
for each representative. Each node in the core can now set its own new ID, 
and send unique new IDs in the range $[1\dots n]$ to the nodes it represents. 
We assume nodes in the core $\C$ take IDs $[1 \dots \nc]$. 

\item \textbf{Fragment ID Initialization.} Each node $u \in V$ forms a singleton fragment with its unique $\mathrm{id}(u)$.

\item \textbf{Obtaining a Leader.} Each initial fragment $i=f(u)$ (which is a singleton at this phase) obtains a 
leader by asking the representative $r(u)$ of node $u$, to select a random Core member $w$ uniformly at random, and declare it as a leader of $i$, $l(i)=w$. 
This is done, in a balanced way, by picking a random permutation and assigning 
leaders according to it, hence every node in $\C$ becomes the 
leader of $O(n_c)$ fragments.
 
%(Michael: by asking its representative $r(u)$ to chose a leader for it.)
\item \textbf{Fragment State Initialization.} Each leader keeps a state ({\em active / frozen / root / waiting}) for each of its fragments. 
The initial state of all fragments is {\em active}.
\end{enumerate}

\paragraph{Phase $b\in \{1 \ldots B\}$ (similar to Boruvka's phases)}

\begin{enumerate}
\item \textbf{Finding \textsc{mwoe}.} 
Each $u \in V$ finds an edge $(u,v)=\mwoe(u)$, and obtains 
$f(v)$ and $l(f(v))$.
\item \textbf{Periphery to Representatives.} 
Each node $u \in V$ sends $(u,v)=\mwoe(u)$, $f(u)$, $l(f(u))$, $f(v)$ and $l(f(v))$ to its representative $r(u)\in \C$.

\item \textbf{Representatives to Leaders.} 
Each representative $w\in \C$, for each fragment $i \in F_{rep}(w)$, 
sends $(u,v)=\mwoe(V^i(w))$, $i$, $f(v)$, and $l(f(v))$ 
to the leader $l(i)$ of $i$.

\item \textbf{Leaders Merge Fragments.} 
Each leader $w \in \C$, for each fragment $i\in F_{lead}(w)$, 
finds $(u,v)=\mwoe(V^i)$ and $\mp(i)=f(v)$, and then executes $\MF(i)$.

%If the result is an \emph{active} state it obtains a new name $newID(i)$ for the fragment.

%\item \textbf{Selecting new leaders.} 
%Each leader $w \in \C$,  for each \emph{active} fragment $i \in F_{lead}(w)$, obtains a new (random) leader for $newID(i)$. 

\item \textbf{Leaders to Representatives.} 
Each leader $w \in \C$, for each \emph{active} fragment $i \in F_{lead}(w)$, 
sends an update message with the new fragment name $newID(i)$,
the new leader node $l(newID(i))$ and the edge to add to the MST, 
to all the representatives of the nodes in $V^i$. 
%If $i \neq newID(i)$, then the fragment $i$ is removed from $F_{lead}(u)$.
If $w \neq l(newID(i))$, then the fragment $i$ is removed from $F_{lead}(w)$. 

\item \textbf{Representatives to Periphery.} 
Each representative $w \in \C$, for each $i \in F_{rep}(w)$, for which 
the update message with $newID(i)$ and $l(newID(i))$ was received, 
forwards it to all the nodes of $V^i(w)$.
\end{enumerate}

%\subsection{\texorpdfstring{$\MF$}{MergFrag} procedure}
\subsection{$\MF$ procedure}
The $\MF$ procedure is the essential part of our algorithm, executed at each 
phase $b$. The procedure is executed by each leader $w\in \C$, for each fragment
$i\in F_{lead}(w)$. For a fragment $i$, its leader maintains a state parameter 
$state(i)\in\{active, frozen, root,waiting\}$. 
Each fragment $i$ attempts to merge with some other fragment $\mp(i)$.
Towards that, the leader of $i$ initiates a merge-request to a leader of the 
fragment $\mp(i)$ (the fragment at the other end of $\mwoe(i)$). 
Since these requests do not have to be reciprocal, merge requests usually form 
a {\em merge-tree}, whose nodes are fragments and whose directed edges represent
merge-requests (see Figure \ref{fig:merge_tree} for illustration). 
In order to minimize the number of merge-request messages sent by fragment 
leaders, we propose to designate, for each set of fragments sharing the same 
leader that attempt to merge with the same target fragment, a \emph{speaker} 
fragment, that will act on behalf of all the fragments in the set, and update 
all of them upon reception of merge-replies.

\begin{figure}[h!]
\centering
\includegraphics[width=11cm]{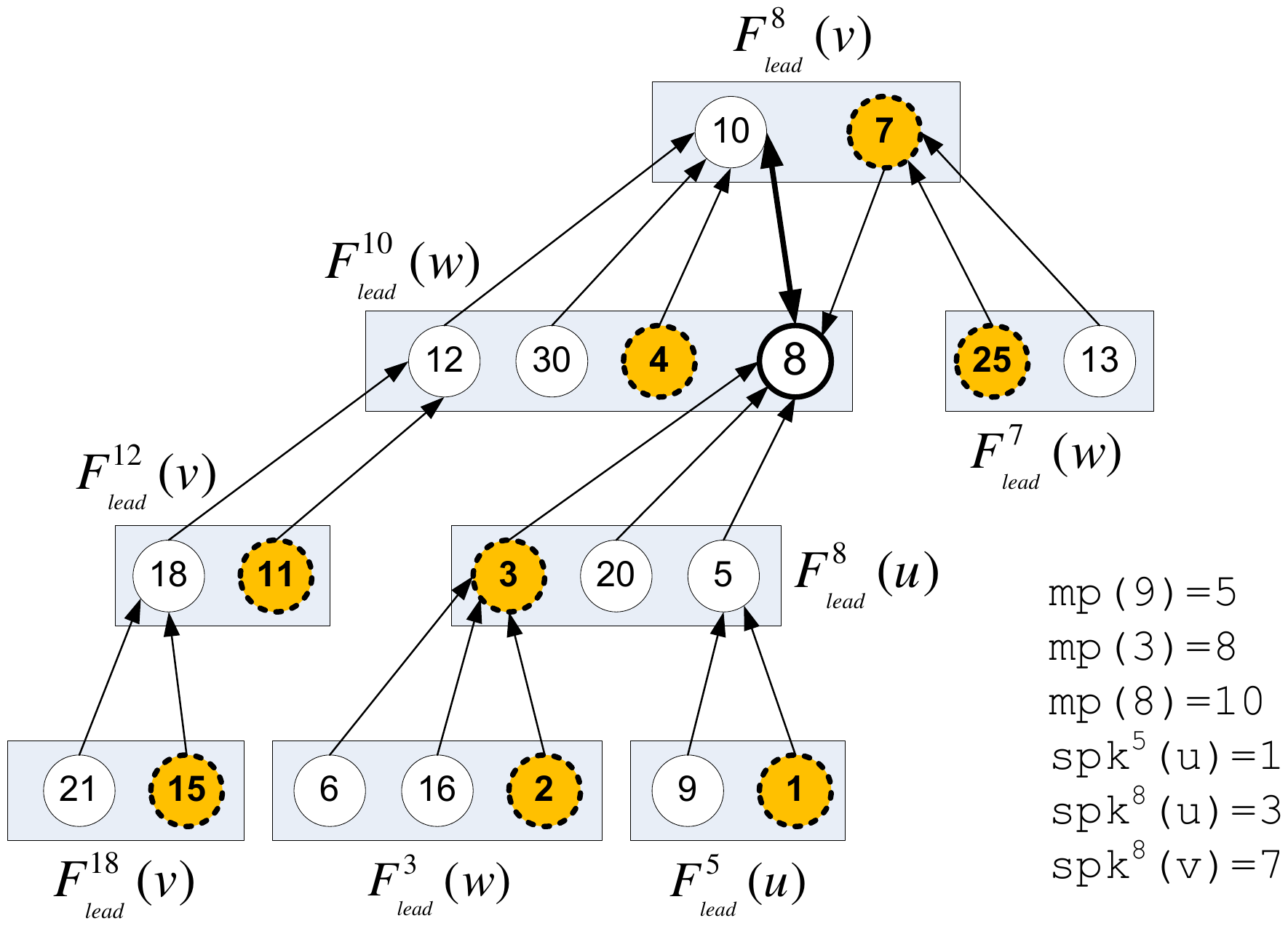}
\caption{Illustration of a fragments merge-tree. An arrow $i\rightarrow j$ 
means that fragment $i$ attempts to merge with fragment $j$, i.e., $j=\mp(i)$. 
The root of the merge-tree is fragment $8$, since it has a reciprocal arrow 
with fragment $10$ and $8<10$.}
\label{fig:merge_tree}
\end{figure}

The root of that tree is a fragment that received a reciprocal merge-request 
(actually, there are two such fragments, so the one with the smaller ID is 
selected as a root). However, since merge-requests are not sent by every 
fragment, only by \emph{speakers}, the root node is detected by the 
\emph{speaker}, and not the root fragment itself (except for the case when
the root is the \emph{speaker}). For example, in Figure \ref{fig:merge_tree}, 
fragment $4$ sends a merge request to fragment $10$, and gets a merge-reply 
with the \emph{next} pointer of $10$, which is $8$ (the \emph{next} pointer 
of fragment $i$ is always set initially to $\mp(i)$). 
Fragment $4$ then realizes that $8$ belongs to $F_{lead}^{10}(w)$, and thus
identifies the reciprocity between $8$ and $10$. Fragment $4$ 
(the \emph{speaker}) then notifies $8$ that it should be the root 
($7$ does not notify $10$ since $8<10$). For a detailed description of 
the root finding procedure see Algorithm $\textsc{FindRoot}(i)$
%\ref{alg:isroot} 
in \ref{app:MST-Pseudocodes}.

When a fragment $i$ that is led by $w$ is in the $active$ state and 
attempts to merge with another fragment ($\mp(i)$), it first tries to find 
the root using the procedure $\textsc{FindRoot}$ (see  \ref{app:MST-Pseudocodes} for the pseudocode). By the end of the 
$\textsc{FindRoot}$ procedure, $i$ may not find a root, in which case its state
will become $frozen$; $i$ may find that the root is another fragment $k$ in 
$F_{lead}^{\mp(i)}(w)$, and then $i$ will notify $k$, but $i$'s state will become 
$frozen$; $i$ may find that it is a root by itself, in which case its state 
will become $root$; and finally, $i$ may be notified by a \emph{speaker} of 
$F_{lead}^{\mp(i)}(w)$ and $i$'s state will become $root$.

%Identified root may be the fragment $i$ itself of some other fragment $k$ in $F_{lead}(w,\mp(i))$. If $i$ is root it sets its own state to $root$, if $i$ found that $k$ should be root, it notifies it, i.e., sets $k$'s state to $root$. In case when $i$ is not $root$ it's state is changed to $frozen$. 

Once a fragment enters the $root$ state, it starts waiting for all the tree 
fragments to send it merge-requests. These merge-requests are sent by each 
fragment, using the pointer-jumping procedure $\pj$ (see  \ref{app:MST-Pseudocodes} for the pseudocode), while it is in the $frozen$ state. 
Once the requests of all the tree fragments reach the $root$ 
(using pointer-jumping), it chooses a new random ID ($newID$) for the fragment, 
among all the fragments in the tree, and a random Core node to be the new leader
($newLead$) for this fragment, and sends this information back to all of them. 
At this point, the merge-tree is considered to be resolved, and all its 
fragments (including the $root$) change their state to $active$. 
The simple state diagram of Algorithm $\MF$ can be found in Figure 
\ref{fig:merge_frag_states}, 
and a detailed pseudocode in \ref{app:MST-Pseudocodes}.

\begin{figure}[h!]
\centering
\includegraphics[width=12.2cm]{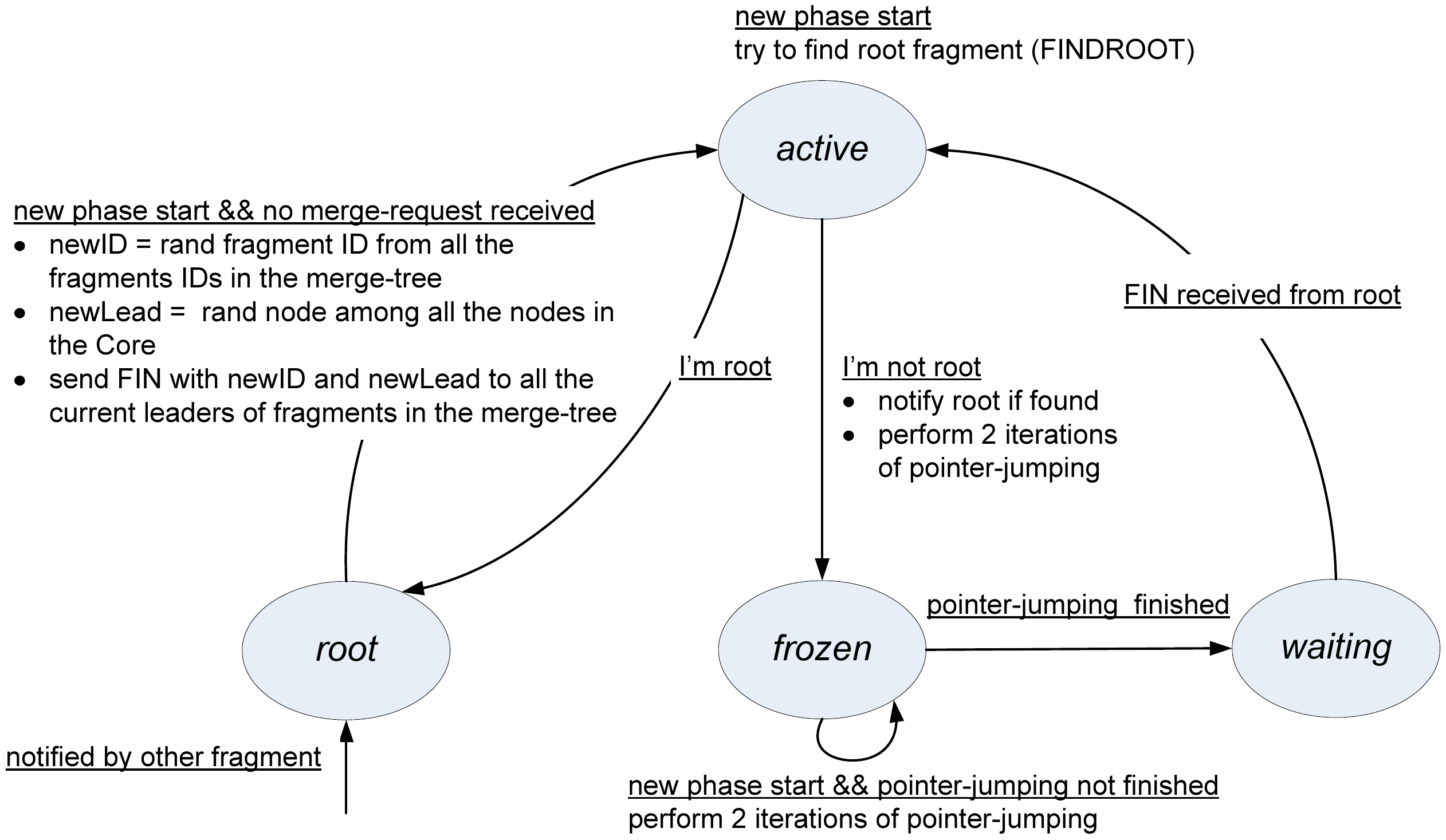}~~
\caption{State diagram of Algorithm $\MF$. The algorithm is executed by every 
leader for every fragment it leads. Each title indicates an event, and the text 
below it is the action performed upon that event.}
\label{fig:merge_frag_states}
\end{figure}

Now we briefly describe the pointer-jumping approach used to resolve 
fragment \emph{merge-trees}. 
Pointer-jumping is a technique (developed in \cite{wyllie1979complexity} 
and often used in parallel algorithms) for contracting a given linked list 
of length $k$, causing all its elements to point to the last element,
in $O(\log k)$ steps. We use the pointer-jumping approach for resolving 
merge-trees, viewing the fragments as the nodes in a linked list,
with each fragment $i$ initially pointing at $\mp(i)$.
%Their ``next" pointers are the destinations of the merge requests they send. 
Each fragment chain can be of length at most $O(n)$, and thus can be 
contracted in $\log n$ rounds, resulting in a $\log n$ time overhead per phase.
In order to overcome this, we use a technique first introduced in 
\cite{lotker06distributed}, called ``amortized pointer-jumping", in which 
the contraction of long chains is deferred to later phases, while in each phase
only a small constant number of pointer-jumps is performed. 
The argument for the correctness of this approach is that if the chain 
(or tree) is large, then the resulting fragment, once resolved, is large 
enough to satisfy the fragment growth rate needed to complete the algorithm 
in $B=O(\log n)$ phases (see Claim \ref{clm:cpmst:logn}). 

\subsection{Correctness of the $\C\P$-MST Algorithm}

We now show that our $\C\P$-MST algorithm is correct, i.e., it results in 
an MST. The following claim shows that the $\MF$ algorithm indeed resolves 
a merge-tree. 

\begin{clm}\label{clm:correct1}
Once a merge-tree becomes {\em active}, all the (old) fragments in the tree 
have the same (new) fragment ID.
\end{clm}

%\begin{clm_correct1}[restated]
%Once a fragment tree becomes {\em active}, all the (old) fragments in the tree 
%have the same (new) fragment ID.
%\end{clm_correct1}

\begin{proof}
The claim follows directly from the description above, and the observation that 
in the pointer-jumping procedure, at every step, at least one more node points 
to the $root$.
Thus, if at some phase the $root$ of the merge-tree does not receive any 
merge-request, then every other fragment in the tree is in $waiting$ state, 
i.e., points to the $root$. Consequently, the $root$ knows all the fragments 
in the tree, and can inform their leaders about the new fragment ID, $newID$, 
and the new leader node $l(newID)$.
\end{proof}

\begin{clm}\label{clm:correct2}
The $\C\P$-MST algorithm emulates Boruvka's MST algorithm 
and thus constructs an MST for the network.
\end{clm}

%\begin{clm_correct2}[restated]
%The $\C\P$-MST algorithm emulates Boruvka's MST algorithm and thus results in MST.
%\end{clm_correct2}
\begin{proof}
In Boruvka's algorithm, fragment merges can be performed in any order. What's important is that a merge between any two fragments will occur if, and only if, they share an edge that is an $\mwoe$, for at least one of the fragments. Since our algorithm satisfies this property, it results in an MST.
\end{proof}

%\subsection{Running time analysis of the \texorpdfstring{$\C\P$}{CP}-MST algorithm}
\subsection{Running time analysis of the $\C\P$-MST algorithm}

We now analyze the running time of Algorithm $\C\P$-MST in a Core-Periphery network.
%that satisfies the Core-Periphery axioms.
%In order to prove this theorem, 
To do that, we analyze each part of the algorithm separately, and prove
%The Theorem follows directly by proving 
the following Claims 
\ref{clm:init_runtime},  \ref{clm:cpmst:logn} and \ref{clm:steps_runtime}.
We start with the initialization phase.

\begin{clm}
\label{clm:init_runtime}
The initialization phase (Phase 0) takes $O(1)$ rounds.
\end{clm}

%\begin{clm_init_runtime}[restated]
%The initialization phase (Phase 0) takes $O(1)$ rounds.
%\end{clm_init_runtime}

\begin{proof}
Due to Axiom $\axiomConvergecast$, the convergecast process employed in step 1 
of Phase 0 requires $O(1)$ rounds.
The renaming step (step 2 of Phase 0) can be done in $O(1)$ rounds, due to Axiom $\axiomClique$. 
The operation of obtaining a leader (step 3 of Phase 0) requires $O(1)$ rounds due to the Axiom $\axiomConvergecast$.
\end{proof}

Now we show that the number of Boruvka phases needed in algorithm $\C\P$-MST 
is $B=O(\log n)$. 

\begin{clm}
\label{clm:cpmst:logn}
Algorithm $\C\P$-MST takes $O(\log n)$ phases, i.e., $B=O(\log n)$. 
\end{clm}
\begin{proof}
The proof is by induction.
Assume that every active fragment $f$ at phase $x \le i$ has size (in nodes) 
$|f| > \min(2^x,n)$.
We show that in the phase $j>i$ at which $f$ becomes active again, 
its size will be at least $\min(2^j,n)$.
In phase $i$, $f$ joins a merge-tree that was created at some phase 
$k\le i$, and according to induction assumption, every fragment in this tree 
has size at least $\min(2^k,n)$.
That tree will be resolved in phase $j$, i.e., after $j-k$ phases. 
Let $D$ be the diameter of the tree in phase $j$. 
Since the algorithm uses pointer jumping with two iterations at each phase, 
it follows that $j-k \le \left\lceil \log D\right\rceil / 2$.
%\begin{align}
%j-k \le \frac{\left\lceil \log D\right\rceil}{2}.
%\end{align} 
The size of the resolved tree is at least $\min(2^k, D)$, since it comprises 
of at least $D$ fragments, each of size at least $\min(2^k,n)$.
Clearly,
$$2^k D ~=~ 2^{k+\log D} ~\ge~ 2^{k+\frac{\left\lceil \log D\right\rceil}{2}} ~\ge~ 2^j,$$
and thus $\left|f\right|\ge \min(2^j,n)$.
%\begin{align}
%\left|f\right|\ge \min(2^j,n).
%\end{align} 
So each active fragment at phase $j$ is of size at least $\min(2^j,n)$. 
If in phase $\lceil\log n\rceil$ there are no active fragments, then the algorithm waits 
for at most $\log n$ time, which is sufficient to resolve any fragments tree, 
and then, the size of the fragment is $\min(2^{2\log n},n) = n$, 
which means that the algorithm has terminated.
\end{proof}

Finally, we analyze the steps performed in phases $b\in[1,\ldots,B]$. 
First, we give the following auxiliary lemma.
The result of this lemma is well known, and its proof is analogous to the proof of Lemma 5.1 in \cite{Upfal_Book_2005}. 
%Since we present this lemma in a slightly more variation, we present it with the proof. 
%, and can be found in \ref{app:MST}.

\begin{lemma}
\label{lem:balls_and_bins}
%When up to $O( n \log n)$ balls are thrown independently and uniformly 
%at random into $\Omega(n)$ bins, the maximum loaded bin has $O(\log n)$ balls 
%with probability at least $1-1/n^8$.
For every real $x > 0$, when up to $k$ balls are thrown independently and uniformly 
at random into at least $w$ bins, the maximum loaded bin has at most $O(k/w + \log x)$ with probability at least $1-w/x^c$, where $c$ is an arbitrary constant.
%
%
% $l$ balls 
%with probability at least $1-w\cdot(3k/(wl))^l$. In particular, we describe two cases: (i) if $k=O(x)$ and $w=\Omega(x)$, then $l=O(\log x)$ w.p. at least $1-(1/x^c)$, where $c$ is an arbitrary constant, and (ii) if $k=O(x)$ and $w=\Omega(\sqrt{x})$, then $l=O(\sqrt{x})$ w.p. at least $1-(1/x^c)$.
\end{lemma}

\begin{proof}
Let $X_i$ be the random variable representing the number of balls in bin $i$. 
For integer $l\ge 0$,
\begin{eqnarray*}
\nonumber
\Pr (X_i \ge l) &\le& {\binom{k}{l}} \cdot {\left(\frac{1}{w}\right)}^l 
~\le~ \frac{k^l}{l!}\cdot \frac{1}{w^l}
~=~ \left(\frac{k}{w}\right)^l \cdot \frac{1}{l!} \\
&\le& \left(\frac{k}{w}\right)^l \cdot \left(\frac{e}{l}\right)^l 
~=~ \left(\frac{ek}{wl}\right)^l
\end{eqnarray*}
For $l=c_1(k/w+\log x)$, we obtain
$$\Pr (X_i \ge c_1(k/w+\log x)) 
~\le~ \left( \frac{e}{c_1}\right)^{c_1(k/w+\log x)} 
~\le~ \frac{1}{x^c}~,$$
where $c=c(c_1)$ is an arbitrary constant.

By taking union bound over all the $w$ bins, we obtain that the probability that any bin has at most $O(k/w + \log x)$ balls with probability of at least $1-w/x^c$.
\end{proof}

\begin{clm}
\label{clm:steps_runtime}
For every phase $b\in[1,\ldots,B]$, the running times of the main steps $1,2$ and $6$ are bounded by $O(1)$, and of steps $3,4$ and $5$ by $O(\log n)$.
Thus, every phase $b$ takes $O(\log n)$ rounds.
\end{clm}

%Now we are ready to proceed with the proof of Claim \ref{clm:steps_runtime}.
\begin{proof}
In step 1, every node sends 
a single message to all its neighbors, so the running time is $O(1)$.
In step 2, each node $u \in V$ sends $\mwoe(u)$ to $r(u)\in \C$ using $\gamma$-convergecast.
By Axiom $\axiomConvergecast$, the running time is $O(\gamma)=O(1)$.
Next, consider step 3. Since the network satisfies Axiom $\axiomClique$, one may assume that $\C$ is a clique.
To derive the running time of this step we have to calculate how many messages 
are sent between a representative $u$ and a leader $v$ in $\C$. It suffices 
to look only at Core edges, since this step involves communication only 
between nodes in $\C$ (representatives and leaders). 
By Theorem \ref{thm:cp_property2}(2), $d_{out}(u)=O(\nc)$, 
and since $\P$ and $\C$ form a $\Theta(1)$-convergecaster, it follows that 
on each edge towards $\P$, $u\in \C$ receives a constant number of 
``representative-requests'' at the initialization phase. The last claim 
implies that $u$ represents $O(\nc)$ nodes, and thus at most $O(\nc)$ fragments.
Hence, every representative node has to send $O(\nc)$ messages, each destined
to a leader of a specific fragment. Since the ``leadership'' on a fragment 
is assigned independently at random to the nodes in $\C$, sending messages 
from representative to leaders is analogous to throwing $O(\nc)$ balls into 
$\nc$ bins. Hence by Lemma \ref{lem:balls_and_bins} with an appropriate constant $c$, the most loaded edge (bin) from a representative $u$ to some leader $v$ 
handles $O(\log \nc)$ messages (balls),
with probability at least $1-1/\nc^8$.
Applying the union bound over all $O(\nc)$ representative nodes and all 
$O(\log n)=O(\log \nc)$ phases of the algorithm, we get that the most loaded 
edge in step 3, is at most $O(\log n)$ with probability 
at least $1-1/n$. Thus, this step takes $O(\log n)$ time.

In step 4, every execution of Procedure $\MF$ 
requires sending / receiving merge-request/reply messages from every leader 
$u\in\C$, for each fragment $i\in F_{lead}(u)$. For each merge-request there is 
exactly one merge-reply, so it suffices to count only merge-requests. 
Moreover, if there are multiple fragments that have the same leader node
and need to send a merge-request to the same other fragment, only one message 
will actually be sent by the \emph{speaker} fragment. 
The last observation implies that every request message sent by a leader 
is destined to a different fragment (i.e., to its leader). As in the analysis 
of the previous step, since ``leadership'' is assigned independently at random 
to the nodes in $\C$, sending messages from leaders to leaders is analogous to 
throwing $A$ balls into $\nc$ bins, where $A$ is the number of fragments that 
the node $u$ leads. Using Lemma \ref{lem:balls_and_bins} with an appropriate constant $c$, $A$ can be bounded with high probability by $O(\sqrt{n})$, 
since up to $n$ fragments (balls) are assigned 
to $\nc=\Omega(\sqrt{n})$ Core nodes (bins).

We now apply Lemma \ref{lem:balls_and_bins} with $O(\sqrt{n})$ balls 
(fragments led by a node), and $\nc=\Omega(\sqrt{n})$ bins (edges towards 
other Core nodes), and conclude that the most loaded edge from a leader $u$ 
to some other leader $v$ carries $O(\log n)$ messages, with probability 
at least $1-1/(\sqrt{n})^8 =1-1/n^4$.
Applying the union bound over all the $O(\sqrt{m})$ leaders, and all 
$O(\log n)$ phases of the algorithm, we get that the most loaded edge 
in the process of sending merge-requests carries at most $O(\log n)$ messages, 
with probability at least $1-1/n^c$. 

The last part of step 4 is when the root fragment 
sends the FIN (``finish'') message to all the merge-tree members. 
at the beginning of phase $j$, the size of each active fragment is at least 
$2^j$ (see Claim \ref{clm:cpmst:logn}) and at most $2^{j+1}$ (as the root 
does not release a tree at phase $j-1$ if it is too large).
Thus, the number of merge-trees resolved at phase $i$ is at most 
$n/2^{i+1}$ (every resolved tree becomes an active fragment 
at the next phase). In case $2^{i+1} \le \sqrt{n}$, it follows from the 
Lemma \ref{lem:balls_and_bins} that a leader node $u\in \C$ has at most 
$O(\sqrt{n}/2^{i+1} + \log n)$ roots (at most $n/2^{i+1}$ balls into at least $\sqrt{n}$ bins). For each root, a leader has to send 
a message, for each fragment in its tree. The number of fragments in the tree 
is bounded by the number of nodes in the tree, which is $2^{i+2}$ 
(this is because the tree becomes an active fragment at the beginning 
of the next phase $j=i+1$, and its size is limited by $2^{j+1}$). 
Thus, a leader has to send 
$O(\sqrt{n}/2^{i+1} + \log n)\cdot 2^{i+2}=O(\sqrt{n}\log n)$ messages. 
Each message is destined to a leader of some fragment, which is located 
at the randomly chosen node in $\C$. So, again, by Lemma \ref{lem:balls_and_bins},
we have that the most loaded edge carries $O(\log n)$ messages 
with high probability. 

In case $2^{i+1} > \sqrt{n}$, Lemma \ref{lem:balls_and_bins} yields 
that a leader $u\in \C$ has at most $O(\log n)$ roots (at most $n/2^{i+1}$ balls into at least $\sqrt{n}$ bins). 
Since a root has to send at most one message to each leader 
(even if the node is a leader of multiple fragments of the tree), 
the total number of messages needed to be sent by a leader is $O(\nc\log n)$. 
Since every message is destined to a random leader, by Lemma \ref{lem:balls_and_bins} we obtain a bound of 
$O(\log n)$ on the maximum edge load, with high probability.

Overall, step 4 takes $O(\log n)$ time.
Step 5 obviously takes the same 
time ($O(\log n)$) as step 3, since it involves the 
transfer of the same amount of information (except in the opposite direction).
Step 6 takes the same time ($O(1)$) 
as step 2, since again it involves
transferring the same amount of information (in the opposite direction).
\end{proof}

%We now summarize the running time of the $\C\P$-MST algorithm:

%\begin{proof}[Proof of Theorem \ref{thm:mst_runtime}]
%The total time bound of $O(\log^2 n)$ rounds follows directly from Claims 
%\ref{clm:init_runtime}, \ref{clm:steps_runtime}, and \ref{clm:cpmst:logn}.
%\end{proof}

We have established the following theorem.
\begin{theorem}
\label{thm:mst_runtime}
%\begin{thm_our_algo}[restated]
On a $\C\P$-network $\GVECP$, Algorithm $\C\P$-MST 
constructs a MST in $O(\log^2 n)$ rounds, with high probability.
%terminates in $O(\log^2 n)$ rounds with high probability.
%\end{thm_our_algo} 
\end{theorem}

\section{Additional Algorithms in Core-Periphery Networks}
\label{sec:more_algorithms}
% !TEX root = core_per.tex

In addition to MST, we have considered a number of other distributed problems 
of different types, and developed algorithms for these problems that can be 
efficiently executed on core-periphery networks.
In particular, we dealt with the following set of tasks, related to matrix 
operations:
(M1) Sparse matrix transposition. 
(M2) Multiplication of a sparse matrix by a vector.
(M3) Multiplication of two sparse matrices.

We then considered problems related to calculating aggregate functions 
of initial values, initially stored one at each node in $V$. 
In particular, we developed efficient algorithms for the following problems:
% the largest $k$ values in the network, 
(A1) Finding the rank of each value, assuming the values are ordered.
(As output, each node should know the rank of the element it stores.) 
(A2) Finding the median of the values.
(A3) Finding the (statistical) mode, namely, the most frequent value.
(A4) Finding the number of distinct values stored in the network. 
Each of these problems requires $\Omega(D)$ rounds on general networks
of diameter $D$. 
We show that on a $\C\P$-network these tasks can be performed in $O(1)$ rounds. 
%(A5) Finding the top $r$ values in each area (i.e., group), assuming that each value belongs to a singe area of a total $O(\sqrt{n})$ areas.

An additional interesting task is defined in a setting where the initial 
values are split into disjoint groups, and requires finding the $r$ largest 
values of each group. This task can be used, for example, for finding 
the most popular headlines in each area of news.
Here, there is an $O(r)$ round solution on a $\C\P$-network.
(The diameter is a lower bound for this task in general networks.)

In all of these problems, we also establish the necessity of all 3 axioms, 
by showing that there are network families satisfying 2 of the 3 axioms, 
for which the general lower bound holds.

%%%%%%%%%%%%%%%%%%%%%%%%%%%%%%%%%%%
\subsection{Technical preliminaries}
%%%%%%%%%%%%%%%%%%%%%%%%%%%%%%%%%%%

A few definitions are in place. Let $A$ be a matrix, in which each entry 
$A(i,j)$ can be represented by $O(\log n)$ bits 
(i.e., it fits in a single message in the CONGEST model).
Denote by $A_{i,*}$ (respectively, $A_{*,i}$) the $i$th row (resp., column) 
of $A$. Denote the $i$th entry of a vector $s$ by $s(i)$.
We assume that the nodes in $\C$ have IDs $[1,\ldots,\nc]$, 
and this is known to all of them.
A square $n\times n$ matrix $A$ with $O(k)$ nonzero entries in 
each row and each column is hereafter referred to as 
an {\em $O(k)$-sparse matrix}.

Our algorithms make extensive use of the following theorem of \cite{Lenzen:2011}.

\begin{theorem}[\cite{Lenzen:2011}]\label{thm:lenzen_stoc2011}
Consider a fully connected network of $\nc$ nodes, where each node is given up to $M_s$ messages to send,
%(the same message to $x$ different destinations are considered as $x$ messages).
 and each node is the destination of at most $M_r$ messages.
There exists an algorithm $\sndmsg$ that delivers all the messages to their destinations in the CONGEST model in $O((M_s+M_r)/\nc)$ rounds with high probability.
\end{theorem}

This theorem provides a time-efficient procedure for messages delivery 
in a core that satisfies Axiom $\axiomClique$. 
Note that the result of the theorem holds with high probability, 
which implies that it exploits a randomized algorithm. 
Nevertheless, our algorithms, presented below, can be considered 
``mostly deterministic'', in the sense that all the {\em decisions} they make 
are deterministic. 
That is, the choices concerning which messages should be send where
during different stages of our algorithms are made deterministically,
and randomness is used only in the information delivery algorithm $\sndmsg$
of Theorem \ref{thm:lenzen_stoc2011}, which is used as a low-level procedure 
for routing messages from sources to destinations over a complete network. 
Hence the time bound of each of our algorithms holds with the same probability 
as those of the calls to $\sndmsg$ used in that algorithm.

In some of our algorithms,
% for calculating aggregate functions
we use the following result on distributed sorting in a complete network, presented in \cite{Lenzen:2013:Sorting}.

\begin{theorem}[\cite{Lenzen:2013:Sorting}]\label{thm:lenzen_podc2013_sorting}
Consider a complete network $G(V,E)$ with node ID's $[1,\ldots,n]$. 
Each node is given $n$ values. For simplicity assume all $n^2$ values are distinct\footnote{This limitation can be eliminated by chaining each value, with the node ID and its order at the node. Thus, each input value becomes unique.}. The following tasks can be performed deterministically in $\Theta(1)$ rounds.
\begin{enumerate}
\item Value learning (VL): Node $i$ needs to learn
the values with indices $[i(n - 1) + 1,\ldots, in]$ according to the total
order of all values.
\item Index Learning (IL): Node $i$ needs to determine the indices of its input (initial) values in the total order of all values.
\end{enumerate}
%
%Even if each node is required to determine the index of its input keys in the total order of the union of all input keys, the
%task still can be solved deterministically in a constant number of
%rounds.
\end{theorem}

\begin{observation}
Theorem \ref{thm:lenzen_podc2013_sorting} can be naturally extended to the case where each node initially holds $O(n)$ keys (instead of exactly $n$).
\end{observation}

\subsection{Matrix transposition (MT)}

Initially, each node in $V$ holds one row of an $O(k)$-sparse matrix $A$ (along with its index).
The {\em matrix transposition (MT)} task is to distributively calculate the matrix $A^T$, and store its rows 
in such a way that the node that stores row $A_{i,*}$ will eventually store 
row $A_{i,*}^T$.
We start with a lower bound.

\begin{theorem}\label{clm:transpose_lower}
Any algorithm for transposing an $O(k)$-sparse matrix on an arbitrary network of diameter $D$ requires $\Omega(D)$ rounds. On a $\C\P$-network, $\Omega(k)$ rounds are required.
\end{theorem}

\begin{proof}
Consider a nonzero entry $A(i,j)$, where $j\neq i$.
Consider the nodes $u$ and $v$ that initially store $A_{i,*}$ and $A_{j,*}$
respectively. Clearly, in any algorithm for MT, $A(i,j)$ should be delivered 
to the node $v$ (which is required to eventually obtain $A_{*,j}=A^T_{j,*}$).
Since the distance $dist(u,v)$ may be as large as the diameter, 
the lower bound is $\Omega(D)$ rounds. 

For a $\C\P$-network, the lower bound on MT is $\Omega(k)$, since there are inputs for which row $A^T_{i,*}$ has $k$ nonzero values, which must be delivered to the node that initially has row $A_{i,*}$. There are $\C\P$-networks whose minimum degree is $1$ (see Figure \ref{fig:examples}(I) for an illustration) and hence delivering $\Omega(k)$ messages will require $\Omega(k)$ communication rounds.
\end{proof}

\paragraph{\bf Algorithm $\algoTranspose$}
MT generating $A^T$ on a $\C\P$-network $\GVECP$.

%\begin{enumerate}

\noindent (1) Each $u\in V$ sends its row (all the nonzero values with their indices in $A$) to its representative $r(u)\in \C$. Now, each representative has $O(\nc)$ rows of $A$ (or, $O(k\nc)$ entries of $A$).

\noindent (2) Each representative sends each entry $A(i,j)$ it has to the node in $\C$ that is responsible for the row $A^T_{j,*}$. Every node in $\C$ is responsible for the rows of $A^T$ indexed $1+(n/\nc)(i-1),\ldots,(n/\nc)i$. (Assume $n/\nc$ is integral.)

\noindent $[*$ Now, each node in $\C$ stores $O(n/\nc)$ rows of $A^T$. $*]$

\noindent (3) Each node $u\in V$ that initially stored the row $i$ of $A$, requests $A^T_{i,*}$ from its representative. The representative gets the row from the corresponding node in $\C$, and sends it back to $u$. 

%\end{enumerate}

\begin{theorem}\label{thm:upper_transpose}
On a $\C\P$-network $\GVECP$, transposing an $O(k)$-sparse matrix can be completed in $O(k)$ rounds with high probability. 
\end{theorem}

%Before we start with the proof, we present the following theorem from \cite{Lenzen:2011}.

\begin{proof}
%[Proof of Theorem \ref{thm:upper_transpose}]
Consider Algorithm $\algoTranspose$ and the $\C\P$-network $\GVECP$.
Step 1 of the algorithm will take $O(k)$ rounds, since each row has up to $k$ nonzero entries and sending one entry takes $O(1)$ due to Axiom $\axiomConvergecast$. 
%($O(1)$ convergecast). 
Now each representative has $O(k\nc)$ values, since it represents up to $O(\nc)$ nodes in $\P$ (due to Axiom $\axiomBoundary$).

In the beginning of Step 2, each representative knows the destination for each of the $A(i,j)$ entries it has (since, by agreement, each node in $\C$ is responsible for collecting entries for specific rows of $A^T$). So, it will send $O(k\nc)$ messages, each one to a specific single destination. Since each node in $\C$ is responsible for $O(n/\nc)$ rows of $A^T$, it will receive $O(k n/\nc)$ messages. Thus, using Axiom $\axiomClique$ and Theorem \ref{thm:lenzen_stoc2011},
the running time is $O(k)$.

At Step 3, a single row (with $O(k)$ nonzero entries) is sent, by each node, to its representative (which takes $O(k)$ time, due to the Axiom $\axiomConvergecast$). Then the requests are delivered to the appropriate nodes in $\C$, and the replies with the appropriate rows of $A^T$ are received back by the representatives. All this takes $O(k)$ rounds, due to Axiom $\axiomClique$ and Theorem \ref{thm:lenzen_stoc2011}. Then the rows of $A^T$ are delivered to the nodes that have requested them. Due to the Axiom $\axiomConvergecast$ this will also take $O(k)$ rounds.
\end{proof}

We now show the necessity of the Axioms $\axiomBoundary$, $\axiomClique$ and $\axiomConvergecast$ for achieving $O(k)$ running time. 

\begin{theorem}\label{thm:necessity-matrix1}
For each $X \in\{B,E,C\}$ there exist a family of $n$-node partitioned networks
${\cal F}_X = \{ G_X(V,E,\C,\P)(n)\}$ 
that satisfy all axioms except $\mathcal{A}_X$,
%that do not satisfy Axiom $\mathcal{A}_X$, but satisfy the other two axioms,
and input matrices of size $n\times n$ for every $n$,
such that the time complexity of \emph{any} matrix transposition (MT) algorithm 
on the networks of ${\cal F}_i$, with the corresponding inputs, is $\Omega(n)$.
\end{theorem}

\begin{proof}
Consider the following cases where in each case, one of the axioms is not satisfied, while the other two are satisfied.

\par\noindent 
Necessity of $\axiomBoundary$: 
Consider the family of dumbbell partitioned networks $D_n$.
As discussed earlier, Axiom $\axiomBoundary$ is violated, while the others hold. 
Let $A$ be a matrix where, at least, the following $n/2$ entries are nonzero 
(assume $n/2$ is even): $A(n/2+1,1), A(n/2+2,2),\ldots,A(n,n/2)$. 
Then we input the rows $A_{1,*}-A_{n/2,*}$ to the nodes in the first star of of $D_n$, 
and rows $A_{n/2+1,*}-A_{n,*}$ to the nodes in the second star of $D_n$. 
Clearly, the entries we specified before are initially located in the second 
star, but they all must be delivered to the first star 
(by the problem definition, an entry $A(i,j)$ should eventually be stored 
in a node that initially has row $A_{j,*}$). 
Since there is only one edge connecting the stars, 
any distributed algorithm for the specified task will take $\Omega(n)$ rounds.

\par\noindent 
Necessity of $\axiomClique$:
Consider the family of sun partitioned networks $S_n$.
As discussed earlier, Axiom $\axiomClique$ is violated, while the others hold. 
The diameter of $S_n$ is $\Omega(n)$, hence, any distributed MT algorithm 
will require $\Omega(n)$ communication rounds 
(due to the lower bound discussed earlier).

\par\noindent 
Necessity of $\axiomConvergecast$:
Consider the family of lollipop partitioned networks $L_n$.
As discussed earlier, Axiom $\axiomConvergecast$ is violated, while the others 
hold. Again, the diameter of $L_n$ is $\Omega(n)$, hence, 
any distributed MT algorithm requires $\Omega(n)$ rounds.
\end{proof}

\subsection{Vector by matrix multiplication (VMM)}

Let $s$ be a vector of size $n$, and $A$ be a square $n\times n$ 
$O(k)$-sparse matrix.
Initially, each node in $V$ holds one entry of $s$ (along with its index), and one row of $A$ (along with its index).
The {\em vector by matrix multiplication (VMM)} task is to distributively calculate the vector $s'=sA$, and store its entries at the corresponding nodes in $V$, such that the node that initially stored $s(i)$ will store $s'(i)$. We start with the lower bound.

\begin{theorem}\label{clm:vertex_by_matrix_lower}
Any algorithm for the multiplication of a vector by an $O(k)$-sparse matrix (VMM) on any network requires $\Omega(D)$ rounds. On a $\C\P$-network, $\Omega(k/\log n)$ rounds are required.
\end{theorem}

\begin{proof}
The $\Omega(D)$ time lower bound for VMM on an arbitrary network follows since in order to obtain $s'(1)$, we need, at least, to multiply $s(1)$ by $A(1,1)$ (assuming $s(1)\neq 0$ and $A(1,1)\neq 0$), which might take $\Omega(D)$ rounds in case $s(1)$ and $A(1,1)$ are located at different nodes $u$ and $v$ at distance $dist(u,v)=D$. 

%TODO... lower bound for a CP-network. Use reduction from the EQ problem.

Now we show that there exists a $\C\P$-network for which the lower bound on VMM is $\Omega(k/\log n)$ rounds.
Consider a $\C\P$-network, as in Figure \ref{fig:examples}(I). Let $u$ be a node in $\P$ (whose degree is $1$). Let $v$ be any other node in $V$. Assume that $u$ initially stores the row $A_{1,*}$,, and the entry $s(1)$, while $v$ stores row $A_{2,*}$ and the entry $s(2)$.
%
%Let us denote the core's nodes as $u$ and $v$. In $O(k)$ rounds $u$ (resp. $v$) can collect all the rows of $A$ and entries of $s$ stored at the nodes of $\P$ connected to $u$ (resp. $v$). So, assuming $n/2$ is integer, after $O(k)$ rounds, $u$ and $v$ have each $n/2$ rows of $A$ and $n/2$ entries of $s$. Assume also an input for which $u$ has rows $A_{1,*},A_{2,*},\ldots,A_{n/2,*}$ and entries $s(n/2+1),s(n/2+2),\ldots,s(n)$, and $v$ has all the remaining rows of $A$ and entries of $s$.

Next, we show a reduction from the well-known \emph{equality problem} (EQ), in which two parties are required to determine whether their input vectors 
$x,y\in \{0,1\}^k$ are equal.
Assuming the existence of a procedure $P$ for our VMM problem, we use it to solve the EQ problem.
Given input vectors $x,y$ for the EQ problem (at $u$ and $v$ respectively), we create an input for the VMM problem in the following way. Node $u$ assigns $A(1,i)=x(i)$, for every $i\in [1,\ldots,k]$ and $s(1)=1$; while node $v$ assigns $A(2,i)=y(i)$ for every $i\in [1,\ldots,k]$ and $s(2)=1$.
All the other entries of $A$ and $s$ are initialized to $0$. It follows that
$s'(i) = \sum_{j=1}^n s(j)A(j,i) = A(1,i)+A(2,i) = x(i)+ y(i)$ for every $i\in [1,\ldots,k]$.
Given the value of $s'(i)$, one can decide whether $x(i)=y(i)$ for every $i\in [1,\ldots,k]$, since clearly,
$x(i)=y(i)$ if'f  $s'(i)\in \{0,2\}$ (and otherwise $s'(i)=1$). 
Notice that the vector $s'$ is stored distributedly in the network, one entry in each node. But the indication to $v$ and $u$, whether all the entries are in $\{0,2\}$, can be delivered in $O(1)$ rounds in the following way. Each node in $\P$ sends its entry of $s'$ to its representative, who checks all the received entries and sends an indication bit to all the other nodes in $\C$. So, every node in $\C$ knows now whether all the entries in $s'$ are in $\{0,2\}$ (actually, we are interested only in the first $k$ entries). Representatives can now inform the nodes in $\P$ they represent in $O(1)$ rounds.
It follows that using procedure $P$, one can solve the EQ problem, which is known to require at least $k$ bits of communication. Therefore, assuming that each message has $O(\log n)$ bits, our problem requires $\Omega(k/\log n)$ communication rounds.
\end{proof}

\paragraph{\bf Algorithm $\algoMatrixByVector$}
VMM on a $\C\P$-network $\GVECP$.

%\begin{enumerate}

\noindent (1) Each $u\in V$ sends the entry of $s$ it has (along with its index) to its representative $r(u)\in \C$ (recall that if $u\in\C$ then $r(u)=u$). 

\noindent (2) $\C$ nodes redistribute the $s$ entries among them, so that the node with ID $i$ stores indices $[1+(n/\nc)(i-1),\ldots, (n/\nc)i]$ (assume $n/\nc$ is integral). 

\noindent (3) Each $u\in V$ sends the index of the row of $A$ it has to 
its representative $r(u)\in \C$. 

\noindent (4) Each representative requests the $s(i)$ entries corresponding to rows 
$A_{i,*}$ that it represents, from the $\C$ node storing it.
%(e.g., representative of the row $A_{7,*}$ ask for $s(7)$). 

\noindent (5) Once getting them, it sends them to the nodes in $\P$ it represents.

\noindent (6) Each $u\in V$ sends the products
% multiplication results 
$\left\{A(i,j)s(i)\right\}_{j=1}^n$ to its representative. 

\noindent (7) Each representative sends each nonzero value $A(i,j)s(i)$ it has
(up to $O(k\nc)$ values) to the representative responsible for $s(j)$, so it can calculate $s'(j)$. 

\noindent (8) Each node $u\in V$, that initially stored $s(i)$, requests $s'(i)$ from its representative. The representative gets the entry from the corresponding node in $\C$ and sends it back to $u$. 

%\end{enumerate}

\begin{theorem}\label{thm:vec_by_matrix_runtime}
On a $\C\P$-network $\GVECP$, the multiplication of a vector by an $O(k)$-sparse matrix (VMM) can be completed in $O(k)$ rounds with high probability. 
\end{theorem}

\begin{proof}
Consider Algorithm $\algoMatrixByVector$ and the $\C\P$-network $\GVECP$.
At Step 1, due to $\axiomBoundary$ 
%(balanced boundary) 
and $\axiomConvergecast$,
%($O(1)$ convergecast) 
each representative will obtain $O(\nc)$ entries of $s$ in $O(1)$ rounds.
For Step 2, we use Theorem \ref{thm:lenzen_stoc2011} with the parameters $M_s=O(\nc)$ and $M_r=O(n/\nc)$, and thus such a redistribution will take $O((\nc+n/\nc)/\nc)=O(1)$ rounds.
At Step 3, due to $\axiomBoundary$ 
%(balanced boundary) 
and $\axiomConvergecast$ 
%($O(1)$ convergecast) 
each representative will obtain $O(\nc)$ row indices of $A$ in $O(1)$ rounds.

For Step 4, we again use Theorem \ref{thm:lenzen_stoc2011} with the parameters $M_s=O(\nc)$ (indices of rows each representative has), $M_r=O(n/\nc)$ (number of entries of $s$ stored in each node in $\C$), and obtain a running time of $O((\nc+n/\nc)/\nc)=O(1)$ rounds for this step.
At Step 5, each representative gets the required elements of $s$, which takes $O(1)$ rounds due to Theorem \ref{thm:lenzen_stoc2011}, and then sends them to the nodes in $\P$ it represents, which also takes $O(1)$ rounds due to $\axiomConvergecast$.
Step 6 takes $O(k)$ rounds, since $A$ has up to $k$ nonzero entries in each row.
Step 7 again uses Theorem \ref{thm:lenzen_stoc2011} with parameters $M_s=O(k\nc)$, $M_r=O(n/\nc)$, and thus its running time is $O(kn/\nc^2)=O(k)$.

At Step 8, a single message is sent by each node to its representative (which takes $O(1)$ rounds due to $\axiomConvergecast$), then the requests are delivered to the appropriate nodes in $\C$, and the replies with the appropriate entries of $s'$ are received back by the representatives. All this takes $O(1)$ rounds, due to Axiom $\axiomClique$ and Theorem \ref{thm:lenzen_stoc2011}. Then the entries of $s'$ are delivered to the nodes that have requested them. By $\axiomConvergecast$ this also takes $O(1)$ rounds.
\end{proof}

The following theorem shows the necessity of the axioms for achieving $O(k)$ running time. 
%The proof of the theorem can be found in \ref{app:proofsVbyM}.

\begin{theorem}\label{thm:necessity-matrix2}
For each $X \in\{B,E,C\}$ there exists a family of $n$-node partitioned networks
${\cal F}_X = \{ G_X(V,E,\C,\P)(n)\}$,
that satisfy all axioms except $\mathcal{A}_X$,
%that do not satisfy Axiom $\mathcal{A}_X$, but satisfy the other two axioms,
and input matrices of size $n\times n$ and vectors of size $n$, for every $n$,
such that the time complexity of \emph{any} algorithm for VMM on the networks of ${\cal F}_X$ with the corresponding-size inputs is $\Omega(n/\log n)$.
\end{theorem}

\begin{proof}
Necessity of $\axiomBoundary$: 
Consider the family of dumbbell partitioned networks $D_n$ (which violates only Axiom $\axiomBoundary$).
Denote the core's nodes as $u$ and $v$. In $O(k)$ rounds $u$ (resp. $v$) can collect all the rows of $A$ and entries of $s$ stored at the nodes of $\P$ connected to $u$ (resp., $v$). So, assuming $n/2$ is an integer, after $O(k)$ rounds, $u$ and $v$ have each $n/2$ rows of $A$ and $n/2$ entries of $s$. Assume also an input, for which $u$ has rows $A_{1,*},A_{2,*},\ldots,A_{n/2,*}$, and entries $s(n/2+1),s(n/2+2),\ldots,s(n)$, and $v$ has all the remaining rows of $A$ and entries of $s$.
We again show a reduction from the EQ problem.
%, in which two parties are required to determine whether their input vectors 
%$x,y\in \{0,1\}^n$ are equal.
%
%DP: {1,2} is cleaner but I prefer using the standard def for EQ
%
%$x,y\in \{1,2\}^n$ are equal\footnote{Originally, the problem is defined for binary vectors but one can easily see that a simple one-to-one function can map our input to the original one.}. 
Assuming the existence of a procedure $P$ for VMM, we use it to solve the EQ problem.
Given input vectors $x,y$ for the EQ problem (at $u$ and $v$ respectively), we create an input for VMM in the following way. Node $u$ assigns $A(i,i)=x(i)+1$ for every $i\in [1,\ldots,n/2]$, and $s(i)=x(i)+1$ for every $i\in [n/2+1,\ldots,n]$; while node $v$ assigns $A(i,i)=y(i)+1$, for every $i\in [n/2+1,\ldots,n]$ and $s(i)=y(i)+1$ for every $i\in [1,\ldots,n/2]$.
All the other entries of $A$ are initialized to $0$, thus $A$ is a diagonal matrix. It follows that
$s'(i) = \sum_{j=1}^n s(j)A(j,i) = s(i)A(i,i) = (x(i)+1)(y(i)+1)$ for every $i$.
%Notice that if $A$ is diagonal, $s'(i)=\sum_{j=1}^n s(j)A(j,i)=s(i)A(i,i)$. 
%Moreover, easy to see that $s'(i)=x(i)y(i), \forall i\in [n/2+1,\ldots,n]$ and $s'(i)=y(i)x(i), \forall i\in [1,\ldots,n/2]$ which means that every entry of the resulting vector $s'$ is the multiplication of the corresponding entries of $x$ and $y$. 
Given the value of $s'(i)$, one can decide whether $x(i)=y(i)$, since clearly,
$x(i)=y(i)$ if'f  $s'(i)\in \{1,4\}$ (and otherwise $s'(i)=2$). 
It follows that, using procedure $P$, one can solve the EQ problem, which in this case requires at least $n$ bits of communication. Hence the VMM problem requires $\Omega(n/\log n)$ communication rounds.
\par\noindent 
The proof of the necessity of $\axiomClique$ and $\axiomConvergecast$ 
is the same as in the proof of Theorem \ref{thm:necessity-matrix1}, 
and is based on the diameter argument.
\end{proof}

\subsection{Matrix multiplication (MM)}

Let $A$ and $B$ be square $n\times n$ matrices, with $O(k)$ nonzero entries in each row and each column.
Initially, each node in $V$ holds one row of $A$ (along with its index), and one row of $B$ (along with its index).
The {\em matrix multiplication (MM)} task is to distributively calculate $C=AB$, and store its rows at the corresponding nodes in $V$, such that the node that initially stored row $i$ of $B$ will store row $i$ of $C$. 
%We start with a claim on the lower bound.

\begin{theorem}\label{clm:matrix_by_matrix_lower}
Any algorithm for $O(k)$-sparse matrix multiplication on any network requires $\Omega(D)$ rounds. On a $\C\P$-network, $\Omega(k^2)$ rounds are required.
\end{theorem}

\begin{proof}
The lower bound for general networks follows since in order to obtain $C(1,1)$ we need, at least, to multiply $A(1,1)$ by $B(1,1)$ (assuming an input in which $A(1,1)\neq 0$ and $B(1,1)\neq 0$), which might take $\Omega(D)$ rounds in case $A(1,1)$ and $B(1,1)$ are located at nodes $u$ and $v$, at distance $dist(u,v)=D$. 

For a $\C\P$-network, consider a network illustrated on Figure \ref{fig:examples}(I), where the degree of a node $u\in\P$ is $1$.
Assume that initially $u$ has row $A_{1,*}$ and $B_{1,*}$ and thus, by problem definition, it has to eventually receive the row $C_{1,*}$.
We show now that $\Omega(k^2)$ messages are required to allow $u$ to obtain $C_{1,*}$. Assume that $\{b_i^j\}$ for $i,j\in[1,\ldots,k]$ are $k^2$ distinct values. Consider an input in which $A(1,i)=1$ for every $i\in [2,\ldots,k+1]$, $B(2,i)=b^1_i$ for every $i\in [1,\ldots,k]$, $B(3,i)=b^2_i$ for every $i\in [k+1,\ldots,2k]$; and so on until $B(k+1,i)=b^k_i$ for every $i\in [k(k-1)+1,\ldots,k^2]$. 
All other entries of $A$ and $B$ are set to $0$.
It is easy to see that $C(1,i)=b^1_i$ for every $i\in [1,\ldots,k]$, $C(1,i)=b^2_i$ for every $i\in [k+1,\ldots,2k]$; and so on until $C(1,i)=b^k_i$ for every $i\in [k(k-1)+1,\ldots,k^2]$. So, in order to obtain the row $C_{1,*}$, $u$ must receive all the $k^2$ values $\{b_i^j\}$, which will take $\Omega(k^2)$ communication rounds.
\end{proof}

\paragraph{\bf Algorithm $\algoMatrixByMatrix$}
$O(k)$-sparse matrix multiplication on a $\C\P$-network $\GVECP$.

%\begin{enumerate}
\noindent (1) Each node in $V$ send its row of $B$ to its representative. \\$[*$ Now, each node in $\C$ has $O(k\nc)$ entries of $B$. $*]$ 

\noindent (2) Nodes in $\C$ redistribute the rows of $B$ among themselves, so that the node with ID $i$ will store rows $1+(n/\nc)(i-1),\ldots (n/\nc)i$ (assuming $n/\nc$ is an integer). \\$[*$ Now, each $u\in \C$ has $O(n/\nc)$ rows of $B$. $*]$

\noindent (3) Each node in $V$ sends its row of $A$ to its representative.
Notice that the row $i$ of $B$ needed to be multiplied only by values of the column $i$ of $A$.

\noindent (4) Each $u\in \C$ sends the values of $A$ it has to the
nodes in $\C$, which hold the corresponding rows of $B$. 
I.e., the value $A(i,j)$ will be sent to the node in $\C$, which holds the row $B_{j,*}$.
\\$[*$ Now all the summands are prepared and distributed across all the nodes in $\C$, and it is left to combine corresponding summands. $*]$

\noindent (5) Each $u\in \C$ sends each of its values to the corresponding node in $\C$, that is responsible for gathering the summands for the values of specific $O(n/\nc)$ rows of the resulting matrix $C$. 

\noindent (6) Each node $u\in V$ that initially stored row $i$ of $B$, requests row $i$ of $C$ from its representative. The representative gets the row from the corresponding node in $\C$, and sends it back to $u$. 
%Easy to see that this step takes $O(k^2)$ rounds.

%\end{enumerate}

\begin{theorem}
On a $\C\P$-network $\GVECP$, the multiplication of two $O(k)$-sparse matrices can be completed in $O(k^2)$ rounds with high probability. 
\end{theorem}

\begin{proof}
Consider Algorithm $\algoMatrixByMatrix$ and the network $G(E,V)$, with a partition $\cp$ that satisfies Axioms $\axiomBoundary$, $\axiomClique$, $\axiomConvergecast$. 

Step 1 takes $O(k)$ rounds, due to the Axiom $\axiomConvergecast$ and the fact that the number of nonzero entries in each row is bounded by $k$. Each node in $\C$ will have $O(k\nc)$ entries of $B$, since it represents $O(\nc)$ nodes in $\P$ due to the Axiom $\axiomBoundary$.
Using Theorem \ref{thm:lenzen_stoc2011} with the parameters $M_s=O(k\nc)$ and $M_r=O(kn/\nc)$, we show that the redistribution, performed at Step 2, takes $O((k\nc+kn/\nc)/\nc)=O(k)$ rounds.

For Step 4, we again use Theorem \ref{thm:lenzen_stoc2011}, with the parameters $M_s=O(k)O(\nc)$ ($O(k)$ values per row), $M_r=O(n/\nc)O(k)$ ($O(k)$ values per column of $A$), and obtain that the running time is $O(k)$.
Note that each node in $\C$ has $O(n/\nc)$ rows (of $B$), with $O(k)$ elements in each. Each row was multiplied by $O(k)$ values received in the previous step ($O(k)$, since each column of $A$ has $O(k)$ nonzero values). Thus, each $u\in \C$ has $O(k^2 n/\nc)$ values that needed to be sent to the appropriate nodes in the next step.

At Step 5, each $u\in \C$ sends each of its values to the corresponding node in $\C$ that is responsible for gathering the summands for the values of specific $O(n/\nc)$ rows of the resulting matrix $C$. Clearly, $M_s=O(k^2n/\nc)$ since each row of $B$ has $O(k)$ different entries, and was multiplied by $O(k)$ different entries of $A$.
Now let's find $M_r$. Each $u\in\C$ is responsible for $O(n/\nc)$ rows of $C$. Thus, e.g., for a row $C_{1,*}$ it needs to receive the following summands: $C(1,1)=\sum A(1,i)B(i,1)$, $C(1,2)=\sum A(1,i)B(i,2)$, ..., $C(1,n)=\sum A(1,i)B(i,n)$. Since the number of nonzero entries in each row of $A$ and $B$ is $O(k)$, the number of nonzero entries in each row of $\C$ is bounded by $O(k^2)$. Thus, for each row of $\C$, a node in $C$ will receive $O(k^2)$ messages. So, $M_r=O(k^2 n/\nc)$ and thus the running time of this step is $O(k^2)$.

At the last step, each node $u\in V$ sends a single message (request for a row) to its representative. This takes $O(1)$ rounds due to $\axiomConvergecast$.
Then, the representative gets the row from the corresponding node in $\C$, and sends it back to $u$. Using Axiom $\axiomClique$ and Theorem \ref{thm:lenzen_stoc2011} with $M_s=O(\nc)$ and $M_r=O(n/\nc)$, delivering the request inside the core takes $O(1)$ rounds. In a similar way, sending the rows inside the core takes $O(k^2)$ rounds. The same amount of time is required to deliver those rows to the nodes in $\P$ that requested them ($O(1)$ per row entry, due to $\axiomConvergecast$, and $O(k^2)$ nonzero entries per row).
\end{proof}

\begin{theorem}\label{thm:necessity-matrix3}
For each $X \in\{B,E,C\}$ there exist a family of $n$-node partitioned networks
${\cal F}_X = \{ G_X(V,E,\C,\P)(n)\}$,
that satisfy all axioms except $\mathcal{A}_X$,
%that do not satisfy Axiom $\mathcal{A}_X$ but satisfy the other two axioms, 
and input matrices of size $n\times n$, for every $n$,
such that the time complexity of \emph{any} MM algorithm for the multiplication for $O(k)$-sparse matrices on the networks of ${\cal F}_X$ with the corresponding inputs is $\Omega(n/\log n)$.
\end{theorem}

\begin{proof}
%Now we show the necessity of the Axioms $\axiomBoundary$, $\axiomClique$ and $\axiomConvergecast$ for achieving $O(k^2)$ running time. Consider the following cases where in each case one of the axioms is not satisfied while the other two satisfied.
Necessity of $\axiomBoundary$: 
Again consider the dumbbell partitioned networks $D_n$.
%As discussed earlier, Axiom $\axiomBoundary$ is violated while the other hold. %
%Let us denote the core's nodes as $u$ and $v$. In order to calculate vector $s'$, row $A_{i,*}$ of $A$ must be multiplied by $s(i)$. Note that there are $\Theta(n)$ rows of $A$ stored at the nodes in $\P$ connected to $u$, and $\Theta(n)$ entries of $s$ stored at the nodes in $\P$ connected to $v$. Assume that the input is arranged in a way that for each row $A(i,*)$ stored at the $u$'s side, the entry $s(i)$ is stored at the $v$'s side. Now, there are $\Omega(n)$ messages must be passed via the edge $(u,v)$ in order to complete the task, and thus $\Omega(n)$ communication rounds.
As done in the proof of Theorem \ref{thm:necessity-matrix2}, we can show a reduction from the EQ problem to the MM problem for sparse matrices. Here we initialize $A$ and $B$ to be diagonal matrices, and the first core node, $u$, will have first half of $A$'s rows and second half of $B$'s rows. Due to the initialization, each entry of the resulting matrix $C$ will be a multiplication of the corresponding entries of $x$ and $y$. Thus, obtaining $C$ will allow us to determine whether $x$ and $y$ are equal. Hence the lower bound for the MM problem is again $\Omega(n/\log n)$ communication rounds.

\par\noindent
The proof for $\axiomClique$ and $\axiomConvergecast$ is the same as in the proof of Theorem \ref{thm:necessity-matrix1}.
\end{proof}

\subsection{Rank Finding (RF)}
Next, we show another set of algorithms that deals with aggregate functions.
Let each node $v\in V$ hold some initial value. The {\em rank finding (RF)} task requires each node to know the position of its value in the ordered list of all the values.

\begin{theorem}\label{clm:rank_lower}
Any algorithm for finding the rank of each value on any network requires $\Omega(D)$ rounds.
% On a $\C\P$-network, the bound is $\Omega(1)$.
\end{theorem}

\begin{proof}
For two nodes $u,v\in V$ to decide whose value is larger, at least one bit of information must travel between them. Thus, the lower bound for this task on any graph is $\Omega(D)$.
\end{proof}

\paragraph{\bf Algorithm $\algoFindRank$} 
RF on a $\C\P$-network $\GVECP$.

\noindent (1) Each node in $V$ sends its initial value to its representative in $\C$.

\noindent (2) Nodes in $\C$ sort all the values they have, and obtain the ranks of those values. \\$[*$ Now, each node $u\in \C$ knows the ranks of the values it had received from the nodes it represents. $*]$

\noindent (3) The ranks are delivered to the appropriate nodes in $\P$.

\begin{theorem}
On a $\C\P$-network $\GVECP$, the RF task can be completed in $O(1)$ rounds with high probability. 
\end{theorem}

\begin{proof}
Consider Algorithm $\algoFindRank$ and the $\C\P$-network $\GVECP$.
The first step takes $O(1)$, due to Axiom $\axiomConvergecast$. 
%($O(1)$ convergecast). 
Now, each representative $u\in \C$ has $O(\nc)$ values, due to Axiom $\axiomBoundary$.
Step 2 is performed in $O(1)$ rounds, due to Theorem \ref{thm:lenzen_podc2013_sorting} and Axiom $\axiomClique$.
The last step takes $O(1)$ rounds, due to $\axiomConvergecast$. 
%($O(1)$ convergecast).
\end{proof}

\begin{theorem}\label{thm:necessity-rank}
For each $X \in\{B,E,C\}$ there exist a family of $n$-node partitioned networks
${\cal F}_X = \{ G_X(V,E,\C,\P)(n)\}$
that satisfy all axioms except $\mathcal{A}_X$,
%that do not satisfy Axiom $\mathcal{A}_X$ but satisfy the other two axioms, 
and input matrices of size $n\times n$ and vectors of size $n$, for every $n$,
such that the time complexity of \emph{any} RF algorithm 
%for finding the rank of each value 
on the networks of ${\cal F}_X$ with the corresponding inputs is $\Omega(n)$.
\end{theorem}

\begin{proof}
%Now we show the necessity of the Axioms $\axiomBoundary$, $\axiomClique$ and $\axiomConvergecast$ for achieving $O(1)$ running time. Consider the following cases where in each case one of the axioms is not satisfied while the other two satisfied.
%
Necessity of $\axiomBoundary$: 
Consider the family of dumbbell partitioned networks $D_n$.
%As discussed earlier, Axiom $\axiomBoundary$ is violated while the others hold.
Consider a sorted list of $n$ values ($a_i\le a_{i+1}$): $a_1,a_2,a_3,...,a_n$. Clearly, every two adjacent values must be compared in any sorting procedure. For example, assuming that some sorting procedure does not compare $a_2$ and $a_3$, and also $a_4=a_3+2$, we can replace $a_2$ with $a_3+1$ and get the same output $a_1,a_2,a_3,...,a_n$, which is now incorrect since $a_2>a_3$.
Denote the two core nodes by $u$ and $v$. Node $u$ and all the nodes in $\P$ that are connected to it, will be assigned values (one per node) with odd indices in the ordered list (i.e., $a_1,a_3,\ldots$). The other $n/2$ values will be assigned to the remaining $n/2$ nodes (one value to one node).
Note that in order to obtain a sorted list, at least the following comparisons must take place:  $a_1$ with $a_2$, $a_3$ with $a_4$, and so on. Hence about $n/2$ pairs of values have to be brought together, while initially they are located at the different sides of the link $(u,v)$. Thus, $\Omega(n)$ messages must be sent over the link $(u,v)$ in order to perform the sorting task. This takes $\Omega(n)$ communication rounds.

The proof for $\axiomClique$ and $\axiomConvergecast$ is the same as in 
Thm. \ref{thm:necessity-matrix1}.
%, and is based on the diameter argument.
%\item $\axiomClique$ is not satisfied, $\axiomBoundary$,$\axiomConvergecast$ satisfied.
%
%Consider the spiked wheel graph $SW_n$ described in the proof of Claim 
%\ref{clm:exist_graph_diam_n}.
%As discussed there, $\axiomBoundary$ and $\axiomConvergecast$ hold
%Clearly, the diameter of $SW_n$ is $\Omega(n)$ (which violates $\axiomClique$) and thus any algorithm for this task will require $\Omega(n)$ communication rounds (due to the lower bound discussed earlier).
%
%\item $\axiomConvergecast$ is not satisfied, $\axiomBoundary$,$\axiomClique$ satisfied.
%
%Consider  the lollipop graph $L_n$ described in the proof of Claim 
%\ref{clm:exist_graph_diam_n}. 
%As discussed there, $\axiomBoundary$ and $\axiomClique$ are satisfied 
%while the Axiom $\axiomConvergecast$ is not. Clearly, the diameter of $L_n$ is $\Omega(n)$ and thus, any algorithm for this task will require $\Omega(n)$ communication rounds (due to the lower bound discussed earlier).
\end{proof}

\subsection{Median Finding (MedF)} 

Let each node $v\in V$ hold some initial value. The {\em median finding (MedF)} task requires each node to know the value, which is the median in the ordered list of all the values in the network.
%
%We start with a claim on the lower bound.

\begin{theorem}\label{clm:median_lower}
Any median finding algorithm on any network requires $\Omega(D)$ rounds.
% On a $\C\P$-network the lower bound is $\Omega(1)$.
\end{theorem}

\begin{proof}
For two nodes $u,v\in V$, if $u$ wants to obtain the median of the initial values, then at least one bit of information must travel between $u$ and $v$ (otherwise $u$ will never know that $v$ exists, while $v$ may even be the median). Since the distance $dist(u,v)$ can be as large as the diameter, the lower bound for this task is $\Omega(D)$ communication rounds.
\end{proof}

\paragraph{\bf Algorithm $\algoMedian$}
MedF on a $\C\P$-network $\GVECP$.

%\begin{enumerate}

\noindent (1) Each node in $V$ sends its initial value to its representative in $\C$.

\noindent (2) Nodes in $\C$ sort all their values and obtain the ranks of those values. \\$[*$ Now, each node $i\in \C$ now knows values with indices $i(\nc - 1) + 1,\ldots i\nc$ according to the total order of all values. $*]$

\noindent (3) The node in $\C$ that has the value with index $n^2/2$ (namely, the median value) sends it to all the nodes in $\C$.

\noindent (4) The median value is delivered to the nodes in $\P$.

%\end{enumerate}

\begin{theorem}
On a $\C\P$-network $\GVECP$, the MedF task can be completed in $O(1)$ rounds with high probability. 
\end{theorem}

\begin{proof}
Consider Algorithm $\algoMedian$ and the $\C\P$-network $\GVECP$.
The first step takes $O(1)$ due to Axiom $\axiomConvergecast$. 
%($O(1)$ convergecast). 
Now, each representative $u\in \C$ has $O(\nc)$ values due to Axiom $\axiomBoundary$. 
%(balanced boundary).
Step 2, is performed in $O(1)$ rounds, due to Theorem \ref{thm:lenzen_podc2013_sorting} and Axiom $\axiomClique$. 
%(clique emulation). 
The last step will take $O(1)$ rounds, due to $\axiomConvergecast$. 
%($O(1)$ convergecast).
\end{proof}

\begin{theorem}\label{thm:necessity-median}
For each $X \in\{B,E,C\}$ there exist a family of $n$-node partitioned networks
${\cal F}_X = \{ G_X(V,E,\C,\P)(n)\}$,
that satisfy all axioms except $\mathcal{A}_X$,
%that do not satisfy Axiom $\mathcal{A}_X$ but satisfy the other two axioms, 
and input matrices of size $n\times n$, and vectors of size $n$, for every $n$,
such that the time complexity of \emph{any} algorithm for finding median of all the initial values on the networks of ${\cal F}_X$, with the corresponding inputs is $\Omega(\log n)$.
\end{theorem}

\begin{proof}
%Now we show the necessity of the Axioms $\axiomBoundary$, $\axiomClique$ and $\axiomConvergecast$ for achieving $O(1)$ running time. Consider the following cases where in each case one of the axioms is not satisfied while the other two satisfied.
Necessity of $\axiomBoundary$: 
Consider the family of dumbbell partitioned networks $D_n$.
As discussed earlier, Axiom $\axiomBoundary$ is violated, while the others hold. 
In constant time, each of the two centers, 
A and B, can learn the inputs of its star. Now, the problem becomes for A and B 
to learn the median of the union of their sets. This operation is known to
require at least $\Omega(\log n)$ communication rounds. More formally,
it is shown in \cite{soltys11}, that the Median function does not admit 
a deterministic protocol using $O(\log^{1-\epsilon} n)$ rounds, and a logarithmic
amount of communication at each round, for any $\epsilon > 0$ 
(even though the total communication complexity of the problem
is known to be only $O(\log n)$ bits).
%The algorithm that does it in $O(\log n)$ bits is based on determining the bits of the median one by one from the most significant one downwards). 
This allows us to conclude the same lower bound in our case too.

The proof for $\axiomClique$ and $\axiomConvergecast$ is the same as in 
Thm. \ref{thm:necessity-matrix1}. 
%and is based on the diameter argument.
%\item $\axiomClique$ is not satisfied, $\axiomBoundary$,$\axiomConvergecast$ satisfied.
%
%Consider the spiked wheel graph $SW_n$ described in the proof of Claim 
%\ref{clm:exist_graph_diam_n}.
%As discussed there, Axioms $\axiomBoundary$ and $\axiomConvergecast$ hold.
%Clearly, the diameter of $SW_n$ is $\Omega(n)$ and thus the running time is $\Omega(n)$ due to the lower bound discussed earlier.
%
%\item $\axiomConvergecast$ is not satisfied, $\axiomBoundary$,$\axiomClique$ satisfied.
%
%Consider  the lollipop graph $L_n$ described in the proof of Claim 
%\ref{clm:exist_graph_diam_n}.
%As discussed before, $\axiomBoundary$ and $\axiomClique$ are satisfied 
%while the $\axiomConvergecast$  is not.
%Clearly, the diameter of $SW_n$ is $\Omega(n)$ and thus the running time is $\Omega(n)$ due to the lower bound discussed earlier.
\end{proof}

\subsection{Mode Finding (ModF)}

Let each node $v\in V$ hold some initial value. The {\em mode finding (ModeF)} task requires each node to know the value (values) that appears most frequently.
%
%We start with a claim on the lower bound.

\begin{theorem}\label{clm:mode_lower}
Any mode finding algorithm on any network requires $\Omega(D)$ rounds.
% On a $\C\P$-network the lower bound is $\Omega(1)$.
\end{theorem}

\begin{proof}
For three nodes $u,v,w\in V$, assume an input, for which the most frequent value appears with frequency $2$, and is initially located at $v$ and $w$, while all the other nodes have other distinct values. Obviously, if $u$ needs to learn the mode, at least one bit of information must travel from $v$ to $u$, since otherwise $u$ will not be aware of the $v$'s existence. Since the distance $dist(u,v)$ can be as large as diameter, the lower bound for this task is $\Omega(D)$ communication rounds.
\end{proof}

\paragraph{\bf Algorithm $\algoFindModes$} 
ModF on a $\C\P$-network $\GVECP$.

%\begin{enumerate}

\noindent (1) Each node in $V$ sends its initial value to its representative in $\C$.

\noindent (2) Nodes in $\C$ perform sorting of all the values they have, and obtain the ranks of those values. So, each node $i\in \C$ now knows values with indices $i(\nc - 1) + 1,\ldots i\nc$, according to the total order of all values. 

\noindent (3) Each node in $\C$ sends to each other node in $\C$:
%\begin{enumerate}
(i) its most frequent value (values) and its frequency,
(ii) its minimum value and its frequency,
(iii) its maximum value and its frequency.
%\end{enumerate}
\\$[*$ The minimum and maximum are needed in order capture the most frequent values that appear at the boundaries of the ranges. $*]$

\noindent (4) Each node in $\C$ finds the most frequent value (values), and delivers it to the nodes in $\P$ it represents. 

%\end{enumerate}

\begin{theorem}
On a $\C\P$-network $\GVECP$, the ModF task can be completed in $O(k)$ rounds with high probability. 
\end{theorem}

\begin{proof}
Consider Algorithm $\algoFindModes$ and the $\C\P$-network $\GVECP$.
The first step takes $O(1)$ due to Axiom $\axiomConvergecast$. 
%($O(1)$ convergecast). 
Now, each representative $u\in \C$ has $O(\nc)$ values due to Axiom $\axiomBoundary$.
%(balanced boundary).
Step 2, is performed in $O(1)$ rounds, due the Theorem \ref{thm:lenzen_podc2013_sorting} and Axiom $\axiomClique$.
% (clique emulation). 
Step 3 takes $O(1)$, due to Axiom $\axiomClique$
% (clique emulation) 
since each node in $\C$ needs to send $O(1)$ values to all the other nodes in $\C$.
The last step will take $O(1)$ rounds due to $\axiomConvergecast$ (and assuming there are $O(1)$ most frequent values).
\end{proof}

\begin{theorem}\label{thm:necessity-mode}
For each $X \in\{B,E,C\}$ there exist a family of $n$-node partitioned networks
${\cal F}_X = \{ G_X(V,E,\C,\P)(n)\}$ 
that satisfy all axioms except $\mathcal{A}_X$,
%that do not satisfy Axiom $\mathcal{A}_X$, but satisfy the other two axioms, 
and input matrices of size $n\times n$ and vectors of size $n$, for every $n$,
such that the time complexity of \emph{any} algorithm for the finding mode on the networks of ${\cal F}_X$, with the corresponding inputs is $\Omega(n/\log n)$.
\end{theorem}

\begin{proof}
%Now we show the necessity of the Axioms $\axiomBoundary$, $\axiomClique$ and $\axiomConvergecast$ for achieving $O(1)$ running time. Consider the following cases where in each case one of the axioms is not satisfied while the other two satisfied.
Necessity of $\axiomBoundary$: 
Consider the family of dumbbell partitioned networks $D_n$.
As discussed earlier, Axiom $\axiomBoundary$ is violated while the other hold. 
Assume such an input that every element 
appears exactly once or twice on each side (for simplicity, assume there are 
$n/4$ types of elements altogether, and some nodes do not have any element). 
Hence, the most frequent element will appear 2, 3 or 4 times in the graph. 
The case where the answer is 2 occurs only when every element appears exactly 
once on every side. This case is easy to check in a constant amount of 
communication between the two centers, so we assume we do this check first, 
and it remains to consider the case where this does not happen. 
It thus remains to decide whether the answer is 3 or 4. 
To do that, A (the center of the first star) defines a set $S_A$ of all 
the elements that appear twice in its star, and B defines a set $S_B$ 
similarly for its side. Now the answer is 4 if'f the sets $S_A$ and $S_B$ 
intersect. Set disjointness has communication complexity $n$, 
so A and B must exchange at least $n$ bits, or, at least $\Omega(n/\log n)$ 
messages. Since these messages all cross the single edge connecting 
the two centers, they require this much time.

The proof for $\axiomClique$ and $\axiomConvergecast$ is the same as in 
Thm. \ref{thm:necessity-matrix1}.
% and is based on the diameter argument.
%\item $\axiomClique$ is not satisfied, $\axiomBoundary$,$\axiomConvergecast$ satisfied.
%
%Consider the spiked wheel graph $SW_n$ described in the proof of Claim 
%\ref{clm:exist_graph_diam_n}.
%As discussed there, Axioms $\axiomBoundary$ and $\axiomConvergecast$ hold.
%Clearly, the diameter of $SW_n$ is $\Omega(n)$ and thus the running time is $\Omega(n)$ due to the lower bound discussed earlier.
%
%\item $\axiomConvergecast$ is not satisfied, $\axiomBoundary$,$\axiomClique$ satisfied.
%
%Consider  the lollipop graph $L_n$ described in the proof of Claim 
%\ref{clm:exist_graph_diam_n}.
%As discussed before, $\axiomBoundary$ and $\axiomClique$ are satisfied 
%while the $\axiomConvergecast$  is not.
%Clearly, the diameter of $SW_n$ is $\Omega(n)$ and thus the running time is $\Omega(n)$ due to the lower bound discussed earlier.
\end{proof}

\subsection{Finding the number of distinct values (DF)}

Let each node $v\in V$ hold some initial value. The {\em number of distinct values finding (DF)} requires each node to know the number of distinct values present in the network.
%
%We start with a claim on the lower bound.

\begin{theorem}\label{clm:distinct_lower}
Any DF algorithm on any network requires $\Omega(D)$ rounds.
% On a $\C\P$-network the lower bound is $\Omega(1)$.
\end{theorem}

\begin{proof}
For two nodes $u,v\in V$, assume all input values in the network are distinct. Obviously, that if $u$ needs to learn the number of distinct values, at least one bit of information must travel from $v$ to $u$, since otherwise, $u$ will not be aware of the $v$'s existence. Since the distance $dist(u,v)$ can be as large as diameter, the lower bound for this task is $\Omega(D)$ communication rounds.
\end{proof}

\paragraph{\bf Algorithm $\algoFindDistinct$} 
DF on a $\C\P$-network $\GVECP$.

%\begin{enumerate}

\noindent (1) Each node in $V$ sends its initial value to its representative in $\C$.

\noindent (2) Nodes in $\C$ perform sorting of all the values they have, and obtain the ranks of those values. 
\\$[*$ Now, each node $i\in \C$ now knows values with indices $i(\nc - 1) + 1,\ldots i\nc$ according to the total order of all values. $*]$

\noindent (3) Each node in $\C$ sends to every other node in $\C$, the number of distinct values and the two border values (min,max). 
\\$[*$ Now, every node in $\C$ is able to find the number of distinct values (for each repeated border value, decrease $1$ from the total count). $*]$

\noindent (4) Each representative delivers the number of distinct values to the nodes in $\P$ it represents. 

%\end{enumerate}

\begin{theorem}
On a $\C\P$-network $\GVECP$, the DF task requires $O(1)$ rounds with high probability. 
\end{theorem}

\begin{proof}
Consider Algorithm $\algoFindDistinct$ and the $\C\P$-network $\GVECP$.
The first step takes $O(1)$ due to Axiom $\axiomConvergecast$. 
%($O(1)$ convergecast). 
Now, each representative $u\in \C$ has $O(\nc)$ values due to the Axiom $\axiomBoundary$. 
%(balanced boundary).
Step 2, is performed in $O(1)$ rounds, due the Theorem \ref{thm:lenzen_podc2013_sorting} and Axiom $\axiomClique$.
% (clique emulation). 
Step 3 takes $O(1)$, due to Axiom $\axiomClique$ 
%(clique emulation) 
since each node in $\C$ needs to send $O(1)$ values to all the other nodes in $\C$.
The last step will take $O(1)$ rounds, due to $\axiomConvergecast$.
\end{proof}

\begin{theorem}\label{thm:necessity-distinct}
For each $X \in\{B,E,C\}$ there exist a family of $n$-node partitioned networks
${\cal F}_X = \{ G_X(V,E,\C,\P)(n)\}$,
that satisfy all axioms except $\mathcal{A}_X$,
%that do not satisfy Axiom $\mathcal{A}_X$ but satisfy the other two axioms, 
and input matrices of size $n\times n$ and vectors of size $n$, for every $n$,
such that the time complexity of \emph{any} algorithm for finding the number of distinct values on the networks of ${\cal F}_X$, with the corresponding inputs is $\Omega(n/\log n)$.
\end{theorem}

\begin{proof}
%Now we show the necessity of the Axioms $\axiomBoundary$, $\axiomClique$ and $\axiomConvergecast$ for achieving $O(1)$ running time. Consider the following cases where in each case one of the axioms is not satisfied while the other two satisfied.
Necessity of $\axiomBoundary$: 
Consider the family of dumbbell partitioned networks $D_n$.
As discussed earlier, Axiom $\axiomBoundary$ is violated, while the others hold. 
Assume that the inputs are taken out of 
a range of $m$ distinct possible values. In constant time, the two star 
centers A and B can collect $m$-bit vectors $x$ and $y$ respectively, 
representing the values that exist in their respective stars. 
The goal is for A and B to decide the number of distinct values in the graph, 
i.e., the number of 1's in the vector $x \vee y$. We show a reduction from 
set disjointness to this problem, hence, the lower bound for set disjointness 
holds for it. Assume we are given a procedure $P$ for our problem. 
We use it to solve set disjointness as follows. 
(1) A computes $|x|$ and informs B. 
(2) B computes $|y|$ and informs A. 
(3) A and B invoke procedure $P$ and compute $|x\vee y|$. 
(4) The answer is ``yes" (the sets are disjoint) iff $|x\vee y| = |x|+|y|$. 
It is easy to verify that the reduction is correct, hence, we get the desired 
lower bound of $\Omega(m/\log n)$ on the number of round required for finding
the number of distinct values.   
\par\noindent
The proof for $\axiomClique$ and $\axiomConvergecast$ is the same as in 
Thm. \ref{thm:necessity-matrix1}.
%, and is based on the diameter argument.
%\item $\axiomClique$ is not satisfied, $\axiomBoundary$,$\axiomConvergecast$ satisfied.
%
%Consider the spiked wheel graph $SW_n$ described in the proof of Claim 
%\ref{clm:exist_graph_diam_n}.
%As discussed there, Axioms $\axiomBoundary$ and $\axiomConvergecast$ hold.
%Clearly, the diameter of $SW_n$ is $\Omega(n)$ and thus the running time is $\Omega(n)$ due to the lower bound discussed earlier.
%
%\item $\axiomConvergecast$ is not satisfied, $\axiomBoundary$,$\axiomClique$ satisfied.
%
%Consider  the lollipop graph $L_n$ described in the proof of Claim 
%\ref{clm:exist_graph_diam_n}.
%As discussed before, $\axiomBoundary$ and $\axiomClique$ are satisfied 
%while the $\axiomConvergecast$  is not.
%Clearly, the diameter of $SW_n$ is $\Omega(n)$ and thus the running time is $\Omega(n)$ due to the lower bound discussed earlier.
\end{proof}

%\subsection{Getting the top  \texorpdfstring{$k$}{k} elements ranked by areas}
\subsection{Get the top $k$ ranked by areas (TopK)}

Let each node $v\in V$ hold some initial value. Assume that each value is assigned to a specific area of a total $\sqrt{n}$ areas. E.g., values may represent news and areas topics, so that each news belongs to a specific topic. Assume also, that each node in $V$ is interested in one specific area.
The {Getting the top $k$ ranked values by areas (TopK)} task is to deliver to each node in $V$ the largest $k$ values, from the area it is interested in.
%
%We start with a claim on the lower bound.

\begin{theorem}\label{clm:k_largest_areas_lower}
Any TopK algorithm on any network requires $\Omega(D)$ rounds. 
On a $\C\P$-network $\Omega(k)$ rounds are required.
\end{theorem}

\begin{proof}
First, let us show the lower bound for any network.
For two nodes $u,v\in V$, assume input in which $u$ is interested in the value initially stored at $v$. Obviously, delivering the value from $v$ to $u$ will take at least $dist(u,v)$. Since the distance $dist(u,v)$ may be as large as diameter, the lower bound on the running time is $\Omega(D)$.

For a $\C\P$-network, the lower bound is $\Omega(k)$, since obviously, there are inputs for which $k$ values must be delivered to a node in $\P$. There are $\C\P$-networks, in which minimum degree is $1$ (see Figure \ref{fig:examples}(I) for an illustration) and hence delivering $\Omega(k)$ messages will require $\Omega(k)$ communication rounds.
\end{proof}

\paragraph{\bf Algorithm $\algoFindKTopAreas$} 
TopK on a $\C\P$-network $\GVECP$.

Without loss of generality, assume that all the values are taken from the range $[1,\ldots,n]$, and each area has its own range for its values (e.g., politics $[1,\ldots,100]$, sports $[101,\ldots,200]$, etc.). 

%\begin{enumerate}

\noindent (1) Each node in $V$ sends its initial value to its representative in $\C$.

\noindent (2) Perform sorting using Theorem \ref{thm:lenzen_podc2013_sorting} and Axiom $\axiomClique$. 
%(clique emulation). 
\\$[*$ Now, each node $i\in \C$ knows values with indices $i(\nc - 1) + 1,\ldots i\nc$ according to the total order of all values. $*]$

\noindent (3) Each node in $\C$ sends the largest $k$ values from each area, to the appropriate node in $\C$, so that each node in $\C$ will be responsible for at most one area (recall that we have $\sqrt{n}$ areas and $\nc=\Omega(\sqrt{n})$). 

\noindent (4) Each representative sends requests for areas (up to $O(\nc)$ areas) requested by nodes it represents. Each request is destined to the specific node in $\C$ responsible for the requested area. Upon request, each node in $\C$ returns the $k$ largest values, for the area it is responsible for, to the requesting nodes. 

\noindent (5) Each representative $u\in\C$ delivers the values to the nodes in $\P$ it represents. 

%\end{enumerate}

\begin{theorem}
On a $\C\P$-network $\GVECP$, the TopK task requires $O(k)$ rounds with high probability. 
\end{theorem}

\begin{proof}
Consider Algorithm $\algoFindKTopAreas$ and the $\C\P$-network $\GVECP$. 
From Theorem \ref{thm:cp_property2}, we know that $\nc=\Omega(\sqrt{n})$, thus we can say that the number of areas is $O(\nc)$.
The first step takes $O(1)$ due to Axiom $\axiomConvergecast$. 
%($O(1)$ convergecast). 
Now, each representative $u\in \C$ has $O(\nc)$ values due to the Axiom $\axiomBoundary$. 
%(balanced boundary).
At Step 2, each node has $O(\nc)$ values (each destined to a specific single node), so $M_s=\nc$. Each node has to receive $M_r=k$ values (more precisely: $M_r=2k$, since it is possible that after the initial sorting, an area is split across two nodes, and each of these two nodes will send up to $k$ values from that area, and the receiving node will have to select the correct $k$). Thus, this step will take (according to Theorem \ref{thm:lenzen_stoc2011}) $O((\nc+k)/\nc)=O(k/\nc)$. 
Step 3 takes $O(k)$, since sending requests will take $O(1)$ (due to the Axiom $\axiomClique$ and Theorem \ref{thm:lenzen_stoc2011} with $M_s=\nc$, $M_r=\nc$), and receiving the desired values will take $O(k)$ (since $M_s=k\nc$ and $M_r=k\nc$).
The last step will take $O(k)$, since each node in $\P$ needs $k$ values, and delivering a single value from $r(u)\in \C$ to $u\in\P$ takes $O(1)$, due to the Axiom $\axiomConvergecast$.
\end{proof}

\begin{theorem}\label{thm:necessity-k-top-area}
For each $X \in\{B,E,C\}$ there exist a family of $n$-node partitioned networks
${\cal F}_X = \{ G_X(V,E,\C,\P)(n)\}$,
that satisfy all axioms except $\mathcal{A}_X$,
%that do not satisfy Axiom $\mathcal{A}_X$ but satisfy the other two axioms, 
and input matrices of size $n\times n$ and vectors of size $n$, for every $n$,
such that the time complexity of \emph{any} algorithm for finding the $k$ top ranked values from a specific area on the networks of ${\cal F}_X$, with the corresponding inputs is $\Omega(k\nc)$.
\end{theorem}

\begin{proof}
Necessity of $\axiomBoundary$: 
Consider the family of dumbbell partitioned networks $D_n$.
As discussed earlier, Axiom $\axiomBoundary$ is violated, while the others hold.
Consider $\sqrt{n}/2$ areas, and assume that all the $k\sqrt{n}/2$ values belonging to these areas are initially located at the nodes of the first star of $D_n$.
Consider a subset of nodes of the second star. Let the subset size be $\sqrt{n}/2$ and each node in this subset is interested in different area, of the areas stored in the first star we mentioned. 
We can see that in order to complete the task, all the $k\sqrt{n}/2$ values have to be sent from the first to the second star, which are interconnected by a single edge. Thus, the running time will be $\Omega(k\sqrt{n})$ rounds.

The proof for $\axiomClique$ and $\axiomConvergecast$ is the same as in 
Thm. \ref{thm:necessity-matrix1}.
\end{proof}

%\section{Discussion and Average Diameter}
%\input{discussion+avg-diam-sections}

%\section{Average Diameter}
%\input{more_results}

%\section{Models}\label{sec:models}

%\input{models}

%\section{Conclusion}\label{sec:conclusion}
%\input{conclusion}
%\clearpage

\section*{References}

{\small
\bibliographystyle{apalike}
%\bibliographystyle{elsarticle-num}
 %\bibliography{literature}

\begin{thebibliography}{}

\bibitem[Adamic, 1999]{adamic1999small}
Adamic, L. (1999).
\newblock The small world web.
\newblock {\em Research and Advanced Technology for Digital Libraries}, pages
  852--852.

\bibitem[Avin et~al., 2012]{Avin2012From}
Avin, C., Lotker, Z., Pignolet, Y.~A., and Turkel, I. (2012).
\newblock From caesar to twitter: An axiomatic approach to elites of social
  networks.
\newblock {\em CoRR}, abs/1111.3374.

\bibitem[Awerbuch, 1987]{awerbuch}
Awerbuch, B. (1987).
\newblock Optimal distributed algorithms for minimum weight spanning tree,
  counting, leader election and related problems.
\newblock In {\em Proc. STOC}, pages 230--240.

\bibitem[Baset and Schulzrinne, 2006]{skype2006}
Baset, S. and Schulzrinne, H. (2006).
\newblock An analysis of the skype peer-to-peer internet telephony protocol.
\newblock In {\em Proc. INFOCOM}, pages 1--11.

\bibitem[Berns et~al., 2012]{berns2012super}
Berns, A., Hegeman, J., and Pemmaraju, S. (2012).
\newblock Super-fast distributed algorithms for metric facility location.
\newblock In {\em Proc. ICALP}, pages 428--439.

\bibitem[Bonanno et~al., 2003]{bonanno2003topology}
Bonanno, G., Caldarelli, G., Lillo, F., and Mantegna, R. (2003).
\newblock Topology of correlation-based minimal spanning trees in real and
  model markets.
\newblock {\em Phys. Rev. E}, 68.

\bibitem[Bonato, 2008]{bonato2008course}
Bonato, A. (2008).
\newblock {\em A course on the web graph}, volume~89.
\newblock AMS.

\bibitem[Borgatti and Everett, 2000]{borgatti2000models}
Borgatti, S. and Everett, M. (2000).
\newblock Models of core/periphery structures.
\newblock {\em Social networks}, 21(4):375--395.

\bibitem[Chen and Morris, 2003]{chen2003visualizing}
Chen, C. and Morris, S. (2003).
\newblock Visualizing evolving networks: Minimum spanning trees versus
  pathfinder networks.
\newblock In {\em Proc. INFOVIS}, pages 67--74.

\bibitem[Chung and Lu, 2006]{chung2006complex}
Chung, F. and Lu, L. (2006).
\newblock {\em Complex graphs and networks}.
\newblock Number 107. AMS.

\bibitem[Dolev et~al., 2012]{dolev2012tri}
Dolev, D., Lenzen, C., and Peled, S. (2012).
\newblock Tri, tri again": Finding triangles and small subgraphs in a
  distributed setting.
\newblock {\em arXiv preprint arXiv:1201.6652}.

\bibitem[Easley and Kleinberg, 2010]{kleinberg-book}
Easley, D.~A. and Kleinberg, J.~M. (2010).
\newblock {\em Networks, Crowds, and Markets - Reasoning About a Highly
  Connected World}.
\newblock Cambridge Univ. Press.

\bibitem[Elkin, 2006]{elkin2006unconditional}
Elkin, M. (2006).
\newblock An unconditional lower bound on the time-approximation trade-off for
  the distributed minimum spanning tree problem.
\newblock {\em SIAM J. Computing}, 36(2):433--456.

\bibitem[Feamster et~al., 2013]{Feamster:SDN}
Feamster, N., Rexford, J., and Zegura, E. (2013).
\newblock The road to sdn.
\newblock {\em Queue}, 11(12):20:20--20:40.

\bibitem[Fujita et~al., 2001]{fujita2001spatial}
Fujita, M., Krugman, P.~R., and Venables, A.~J. (2001).
\newblock {\em The spatial economy: Cities, regions, and international trade}.
\newblock MIT press.

\bibitem[Gallager et~al., 1983]{GHS}
Gallager, R., Humblet, P., and Spira, P. (1983).
\newblock A distributed algorithm for minimum-weight spanning trees.
\newblock {\em ACM Trans. on Programming Lang. and Syst.}, 5:66--77.

\bibitem[Garay et~al., 1998]{garay1998sublinear}
Garay, J.~A., Kutten, S., and Peleg, D. (1998).
\newblock A sublinear time distributed algorithm for minimum-weight spanning
  trees.
\newblock {\em SIAM J. Computing}, 27:302--316.

\bibitem[Hojman and Szeidl, 2008]{hojman2008core}
Hojman, D. and Szeidl, A. (2008).
\newblock Core and periphery in networks.
\newblock {\em J. Economic Theory}, 139(1):295--309.

\bibitem[Holme, 2005]{holme2005core}
Holme, P. (2005).
\newblock Core-periphery organization of complex networks.
\newblock {\em Physical Review E}, 72:046111.

\bibitem[Jung et~al., 2012]{jung2012distributed}
Jung, K., Kim, B., and Vojnovic, M. (2012).
\newblock Distributed ranking in networks with limited memory and
  communication.
\newblock In {\em Proc. ISIT}, pages 980--984.

\bibitem[Krugman, 1991]{krugman-1991}
Krugman, P. (1991).
\newblock {Increasing Returns and Economic Geography}.
\newblock {\em J. Political Economy}, 99(3):483--499.

\bibitem[Kutten and Peleg, 1998]{kutten1998fast}
Kutten, S. and Peleg, D. (1998).
\newblock Fast distributed construction of small k-dominating sets and
  applications.
\newblock {\em J. Algorithms}, 28:40--66.

\bibitem[Lenzen, 2013]{Lenzen:2013:Sorting}
Lenzen, C. (2013).
\newblock Optimal deterministic routing and sorting on the congested clique.
\newblock In {\em Proc. PODC}, pages 42--50.

\bibitem[Lenzen and Wattenhofer, 2011]{Lenzen:2011}
Lenzen, C. and Wattenhofer, R. (2011).
\newblock Tight bounds for parallel randomized load balancing.
\newblock In {\em Proc. STOC}, pages 11--20.

\bibitem[Liang et~al., 2006]{Liang2006842}
Liang, J., Kumar, R., and Ross, K.~W. (2006).
\newblock The fasttrack overlay: A measurement study.
\newblock {\em Computer Networks}, 50:842 -- 858.

\bibitem[Lotker et~al., 2005]{lotker2005minimum}
Lotker, Z., Patt-Shamir, B., Pavlov, E., and Peleg, D. (2005).
\newblock Minimum-weight spanning tree construction in o(log log n)
  communication rounds.
\newblock {\em SIAM J. Computing}, 35(1):120--131.

\bibitem[Lotker et~al., 2006]{lotker06distributed}
Lotker, Z., Patt-Shamir, B., and Peleg, D. (2006).
\newblock Distributed {MST} for constant diameter graphs.
\newblock {\em Distributed Computing}, 18(6):453--460.

\bibitem[Lynch, 1995]{lynch-book}
Lynch, N. (1995).
\newblock {\em Distributed Algorithms}.
\newblock Morgan Kaufmann.

\bibitem[MacCormack, 2010]{maccormack2010architecture}
MacCormack, A. (2010).
\newblock The architecture of complex systems: do core-periphery structures
  dominate?
\newblock In {\em Proc. Acad. Management}, pages 1--6.

\bibitem[Mitzenmacher and Upfal, 2005]{Upfal_Book_2005}
Mitzenmacher, M. and Upfal, E. (2005).
\newblock {\em Probability and Computing: Randomized Algorithms and
  Probabilistic Analysis}.
\newblock Cambridge Univ. Press.

\bibitem[Nesetril et~al., 2001]{Boruvka:1926}
Nesetril, J., Milkova, E., and Nesetrilova, H. (2001).
\newblock Otakar boruvka on minimum spanning tree problem translation of both
  the 1926 papers, comments, history.
\newblock {\em Discrete Mathematics}, 233(1 - 3):3 -- 36.

\bibitem[Newman, 2010]{newman2010networks}
Newman, M. (2010).
\newblock {\em Networks: an introduction}.
\newblock Oxford Univ. Press.

\bibitem[Patt-Shamir and Teplitsky, 2011]{patt2011round}
Patt-Shamir, B. and Teplitsky, M. (2011).
\newblock The round complexity of distributed sorting.
\newblock In {\em Proc. PODC}, pages 249--256.

\bibitem[Peleg, 2000]{peleg-book}
Peleg, D. (2000).
\newblock {\em Distributed Computing: A Locality-Sensitive Approach}.
\newblock SIAM.

\bibitem[Peleg and Rubinovich, 2000]{peleg2000near}
Peleg, D. and Rubinovich, V. (2000).
\newblock A near-tight lower bound on the time complexity of distributed
  minimum-weight spanning tree construction.
\newblock {\em SIAM J. Computing}, 30:1427--1442.

\bibitem[Soltys, 2011]{soltys11}
Soltys, K. (2011).
\newblock The hardness of median in the synchronized bit communication model.
\newblock In {\em Proc. TAMC}, pages 409--415.

\bibitem[Wyllie, 1979]{wyllie1979complexity}
Wyllie, J. (1979).
\newblock {\em The complexity of parallel computations}.
\newblock Technical report. Dept. of Computer Science, Cornell University.

\end{thebibliography}

}

\clearpage
\pagenumbering{roman}
\appendix
\centerline{{\bf\large APPENDIX}}

\section{Pseudocodes for Section \ref{sec:mst}}
\label{app:MST-Pseudocodes}
\AppendixPseudocodes

\end{document}